\documentclass[12pt, a4paper]{amsart}
\usepackage[in]{fullpage}
\usepackage{amssymb}
\usepackage{amsmath}
\usepackage{amsthm}
\usepackage{bbm}
\usepackage{mathrsfs}
\usepackage{mathabx}
\usepackage{enumitem}
\usepackage[colorlinks=true,urlcolor=blue]{hyperref}
\usepackage{xcolor}
\usepackage{graphicx}
\usepackage{pst-node}
\usepackage{tikz-cd} 
\usepackage{lmodern}
\usepackage{multicol}
\usepackage{lipsum}
\usepackage{atbegshi}
\usepackage{mwe}
\usepackage{xcolor}
\usepackage{graphicx}
\usepackage{amssymb}
\usepackage{here}
\usepackage{color}
\usepackage{framed}
\usepackage{wrapfig}
\usepackage[none]{hyphenat}
\usepackage{textcomp}
\usepackage{physics}
\usepackage{caption}
\usepackage{comment}
\usepackage{tabu}
\usepackage{amsfonts}
\usepackage{booktabs}
\usepackage{siunitx}
\usepackage{mathtools}
\usepackage{dsfont}
\usepackage{setspace}
\usepackage{changepage}
\usepackage{array}
\usepackage{multirow}
\usepackage{soul}
\usepackage{pdfpages}
\usepackage[backend=biber]{biblatex}
\renewbibmacro{in:}{}
\usepackage{listings}
\usepackage{color}

\definecolor{dkgreen}{rgb}{0,0.6,0}
\definecolor{gray}{rgb}{0.5,0.5,0.5}
\definecolor{mauve}{rgb}{0.58,0,0.82}

\numberwithin{equation}{section}

\newtheorem{theorem}{Theorem}[section]
\newtheorem{lemma}[theorem]{Lemma}
\newtheorem{proposition}[theorem]{Proposition}

\newtheorem{corollary}[theorem]{Corollary}

\newtheorem*{qst}{Question}
\newtheorem*{strategy}{Strategy}

\newtheorem{definition}[theorem]{Definition}
\newtheorem{notation}[theorem]{Notation}

\newtheorem{remark}[theorem]{Remark}

\newtheorem{conjecture}[theorem]{Conjecture}
\newtheorem{observation}[theorem]{Observation}
\newtheorem{fact}[theorem]{Fact}
\newtheorem{property}[theorem]{Property}
\numberwithin{equation}{section}
\usepackage{mathtools}
\newcommand{\defeq}{\vcentcolon=}

\newcommand{\req}{\overset{!}{=}}
\newcommand{\minus}{\scalebox{1.0}[1.0]{$-$}}

\normalsize

\def\XXint#1#2#3{{\setbox0=\hbox{$#1{#2#3}{\int}$ }
\vcenter{\hbox{$#2#3$ }}\kern-.6\wd0}}

\newcommand{\x}{\mathbf{x}}

\newcommand{\la}{\ensuremath{\mathfrak{l}\text{ }{}}}
\newcommand{\A}{\ensuremath{\alpha{}}}
\newcommand{\RNum}[1]{\uppercase\expandafter{\romannumeral #1\relax}}
\newcommand{\estimates}{\overset{\scriptscriptstyle\wedge}{=}}

\newenvironment{psm}
  {\left(\begin{smallmatrix}}
  {\end{smallmatrix}\right)}
\DeclareMathOperator{\spn}{span}

\begin{document}
\pagenumbering{roman}
\title{Lattices: from roots to string compactifications}
\author{366142}
\includepdf[pages={1}]{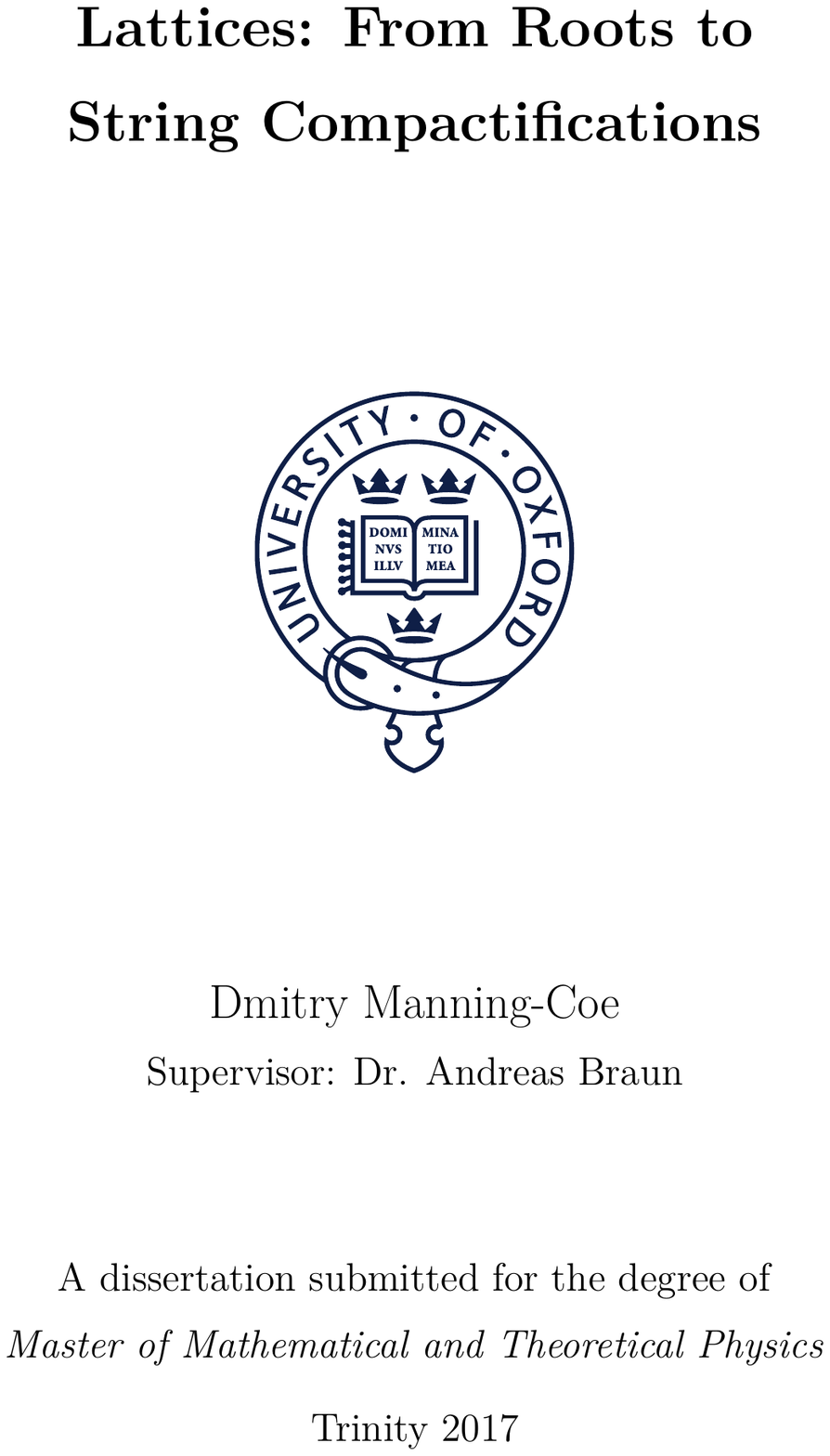}
\includepdf[pages={1}]{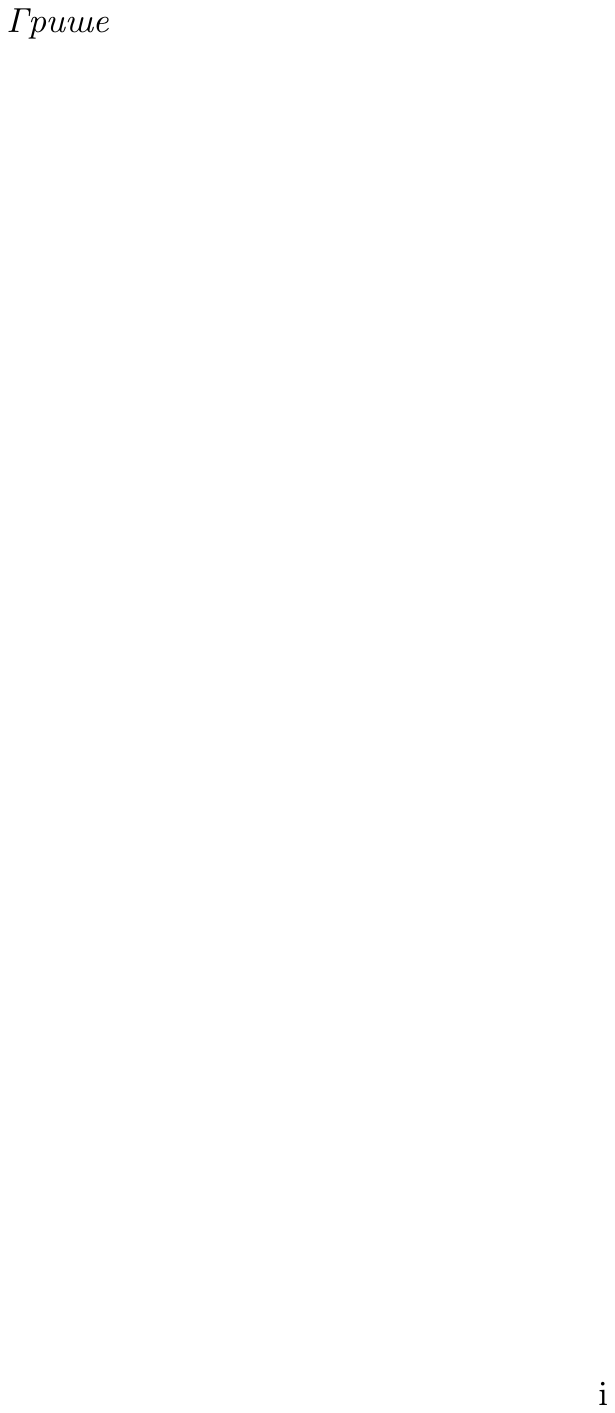}
\setcounter{page}{2}
\pagenumbering{roman}
\includepdf[pages=1]{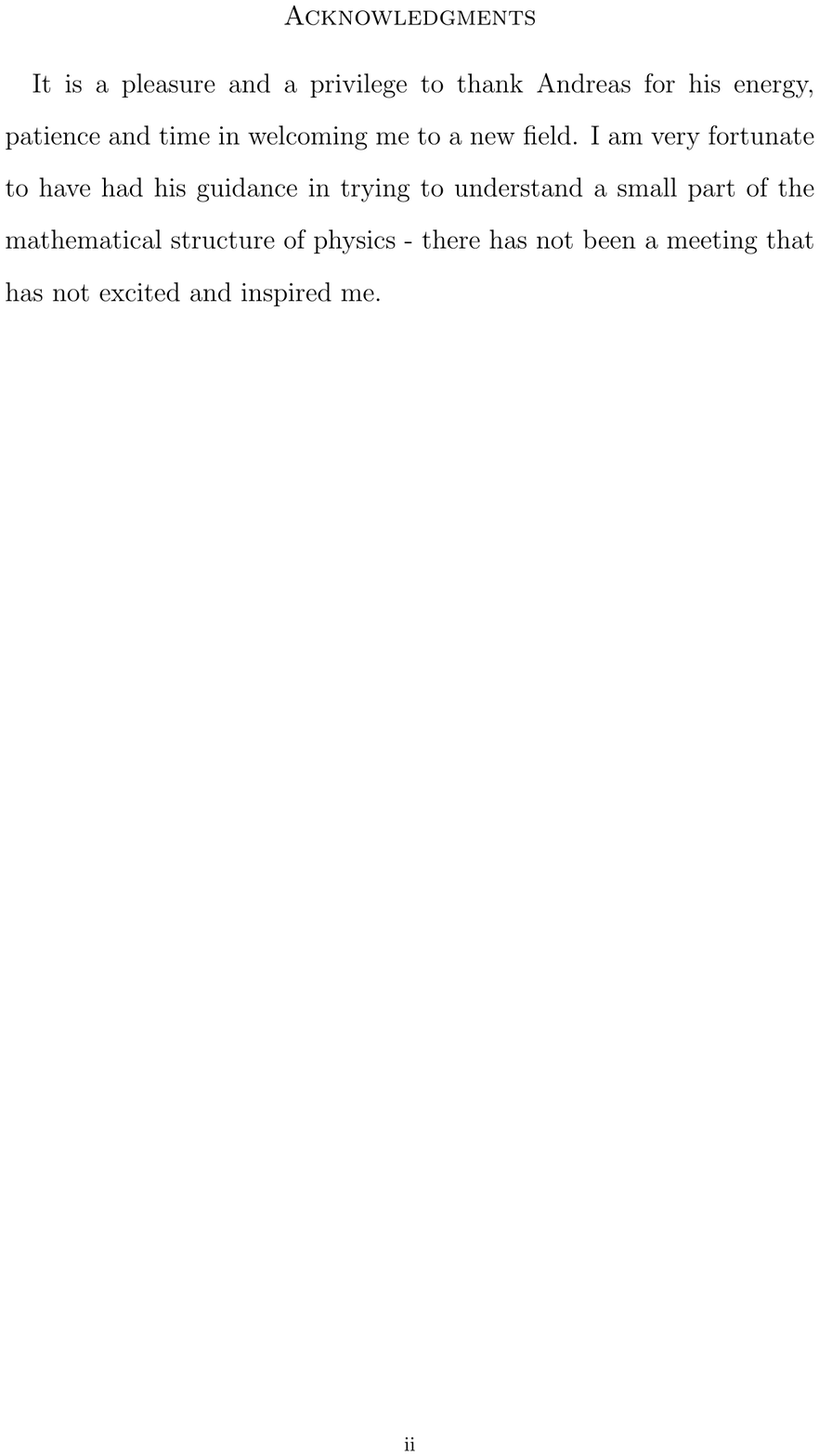}
\doublespacing
\pagebreak
\hspace{0pt}
\vfill
\begin{center}
\setcounter{page}{3}
{\fontfamily{qbk}\selectfont{Abstract}}
\end{center}
In this dissertation we build on the work of \cite{1312}, \cite{1401} and attempt to answer the question of whether non-abelian gauge groups occur in all F-theory models with Picard number 20 and with a Calabi-Yau four-fold $CY_4=K3\times K3$ 
 \cite{Dasgupta:1999ss}. To do so we employ the Kneser-Nishiyama method \cite{Nishiyama} to study the properties of elliptic fibrations on a $K3$ surface. In \cite{1312}, it was found that there are 34 lattices which must be considered and the question was partially answered for one of them. Here we use results in the sphere packing problem to completely answer this question for 21 cases, and show that a negative answer to this question would constitute a new optimal lattice sphere packing in dimension 18.
\vfill
\hspace{0pt}
\pagebreak
\tableofcontents

\clearpage
\pagenumbering{arabic}
\setcounter{page}{1}
\section{Introduction}
It has become canonical to define a quantum field theory algebraically by expressing the symmetries of it's Lagrangian as a Lie group. No doubt, this formulation has become standard because it has been so strikingly succesful. At first glance, it is similarly striking how different this mathematical structure is from the geometric framework employed by general relativity. It is the promise of string theory that these two mathematical traditions can be combined into a viable physical theory. The primary aim of this dissertation is to make a modest contribution to understanding the properties of such a theory, by asking the following question:
\begin{qst}
Does non-abelian gauge symmetry occur generically in an interesting and tractable class of string theories?
\end{qst}
A lattice is by intuition a geometric object, and by definition an algebraic one.\footnote{The concern of this dissertation is \textit{group} lattice theory, as distinct from \textit{order} lattice theory.} It is natural then that this question can be formulated in lattice theoretic terms. We shall make this translation, and in so doing discover a connection between the sphere packing problem and lattice embeddings. \\
\indent We begin on the side of algebra, showing how the Serre construction can be used to recover the entire Lie algebra from it's root lattice. This then puts us into a position to introduce the particular string theory with which we are concerned. In turn, this allows us to state more precisely the physical motivation of this project.\\
\indent Although our motivation is physical, our focus is on the mathematics. A discussion of the fundamentals of lattice theory leads into the classification of non-unimodular lattices. The  classic work of Nikulin on this subject \cite{Nikulin} provides the means to make the translation we require through the Kneser-Nishiyama \cite{Nishiyama} method. The first step in this method is the computation of 34 lattices which was begun in \cite{1312} and which is completed here. \\
\indent The second step requires the introduction of the rank 24 integral lattices with minimal determinant, of which there are 24. We present the constructions of 23 of them from glue theory. The twenty-fourth, the Leech lattice, is a miraculous object. An overview of relevant results in the sphere packing problem is given before it's construction in order to better appreciate it's optimality.\\
\indent We are then able to demonstrate that the lattice theoretic formulation of our question is exactly a restricted form of the sphere packing problem. This in turn allows us to establish conclusively that the answer to this question is affirmative for 21 of the 34 cases that must be considered. The results we will require to show this make it plausible that the conjecture is true for all 34 cases. Moreover, we show that the falsity of the conjecture  would entail a new result in lattice sphere packings that does not easily fit into the current known structure in dimensions less than 24. \\
\indent Lattice theory sprawls across a large swathe of modern mathematics: it uses and extends results in algebra, geometry and analysis. This forces a short work to be necessarily selective. The aim here has been to sacrifice comprehensiveness for clarity. This has led to some painful omissions: we do not give the connection to analysis through theta functions \cite{Thetaintro},\cite[Chapter~2]{CS}, we have had to omit Venkov's elegant proof  \cite{Venkov} that the Niemeier lattices are exhaustive and the discussion of Lie algebras is as short as coherence allows. \\
\indent This dissertation is self-contained. This provides an oppurtunity to offer an introduction to lattice theory that is accesible to physicists without a background in abstract algebra.  Lattice theory deserves to be wider known, and it is hoped that this dissertation is an easy to way to become familiar with it. 
\clearpage

\vspace{15mm}

\section{Lie groups, Lie algebras and physics.}
In this section, we review the classification of the finite-dimensional semi-simple Lie algebras and show how we may characterise a Lie algebra through a lattice associated to it - it's \textit{root lattice}.
The results are standard in Lie theory, so we will often omit proofs to focus on the connection to lattice theory.
\subsection{Lie groups and Lie algebras}\hfill

A Lie group can be thought of as a group of continuous symmetries. Ubiquitous examples are rotational and gauge symmetry. Formally:
\begin{definition}[Lie group]
A Lie group, $G$, is a group which is also a smooth differentiable manifold in which the group binary operations of multiplication $\circ$, and inversion $i$: 
$$ \circ : G\times G\rightarrow G $$
$$i : G\rightarrow G$$
are differentiable maps.
\end{definition}
\vspace{-1mm}
It is simpler, however, to work with a linear object associated with a Lie group - an algebra termed it's \textit{Lie algebra}. Every Lie group has an associated Lie algebra, but a Lie algebra can be defined independently of a Lie group as:
\begin{definition}[Lie algebra]
A Lie algebra $\mathfrak{l}$ is a vector space over a field $K$ with a binary operation $[\cdot,\cdot]$ satisfying the following conditions:\\
1. $\text{ }[ax+by,z]=a[x,z]+b[y,z]$ and $[z,ax+by]=a[z,x]+b[z,y],\text{  }\forall a,b \in K\text{ and }\forall x,y,z\in \mathfrak{l}$ (Bilinearity)\\
2. $\forall x,y\in \mathfrak{l},\text{ } [x,y]=-[y,x]$ (Anti-commutativity)\\
3. $\forall x,y,z\in \mathfrak{l},\text{ }[x,[y,z]]+[z,[x,y]]+[y,[z,x]]=0$ (Jacobi Identity)
\end{definition}
The correspondence between the group and the algebra is the \textit{exponential map}, which takes every element of a Lie algebra to an element in the Lie group:
\begin{equation}
Exp: \mathfrak{g}\rightarrow G
\end{equation}
\begin{remark}
The Lie algebra corresponds only to the \textit{connected component} of $G$ at the identity. It is not the case that the entire group structure can be recovered from the Lie algebra, as the group may have non-trivial topological structure. The Lie algebra is essentially the `linearised' group manifold at the identity. 
\end{remark}

	Physical theories are often defined on a vector space, and so we require a map from the elements of the Lie algebra to objects which can act on a vector space. These objects are elements of the general linear group, matrices, and the map is a \textit{representation}:

\begin{definition}[Representation of a Lie algebra]
A representation of a Lie algebra $\mathfrak{l}$ on a finite dimensional vector space $V$ over a field $K$ is a map $r$ from $\mathfrak{l}$ to the group of linear maps on $V$ over $K$:
$$r : \mathfrak{l}\rightarrow GL(V,K)$$
which preserves the Lie bracket. That is:
$$\forall S,T\in \mathfrak{l}\text{   }r[S,T]=[r(S),r(T)]$$
\end{definition}
\vspace{-3mm}
A particularly important class of representations are the \textit{irreducible representations}, or \textit{irreps}:
\begin{definition}[Irreducible representation]
An irreducible representation on a vector space $V$ is one with no invariant subspaces that are not ${0}$ or $V$.
\end{definition}

\begin{definition}[Invariant subspace]
An invariant subspace $U$ of a representation $r$ on a vector space $V$ is one for which $U\subseteq V$ and
\begin{equation}
\text{   }r(l)u\in U,\text{  }\forall u\in U\text{ }\mathrm{and}\text{ }\forall l\in \mathfrak{l}
\end{equation}
\end{definition}
If a representation is not irreducible it is said to be \textit{reducible} and in this case the linear maps (matrices) corresponding to that representation will be in block-diagonal form. 

\subsubsection{Physical application}
The important point for applications is that the symmetries of a physical theory 
 are the elements of its Lie group and we can study them by studying a Lie algebra corresponding to that group. Irreducible representations of the Lie algebra correspond to particles. Each irrep contains weight vectors and these constitute the spectrum of quantum states of the particle. The Lie group-Lie algebra correspondence is summarised as:
\usetikzlibrary{matrix}
\begin{figure}[h]

\begin{tikzpicture}
  \matrix (m) [matrix of math nodes,row sep=3em,column sep=4em,minimum width=2em] {
     G & \tilde{G} \\
     \mathfrak{l} & \tilde{\mathfrak{l}} \\};
  \path[-stealth]
    (m-2-1) edge node [left] {$Exp$} (m-1-1)
    (m-1-1) edge node [above] {$\mathcal{F}$} (m-1-2)       
    (m-2-1) edge node [above] {$\mathcal{M}$}  (m-2-2)
    (m-2-2) edge node [right] {$\widetilde{Exp}$} (m-1-2);
     
\end{tikzpicture}
\caption{Lie algebra-Lie group correspondence. Here $\mathcal{F}$ is a Lie group homomorphism, and $\mathcal{M}$ is a Lie algebra homomorphism.} \label{fig: Fig.1}
\end{figure}

\subsection{Lie algebra classification}\hfill

Given that Lie algebras correspond to the symmetries of a physical theory, a natural question to ask is ``What kinds of Lie algebras are there?". The answer to this question for a wide class of Lie algebras, the \textit{finite-dimensional semi-simple Lie algebras}, is contained in the celebrated classification of the \textit{simple Lie algebras}, which is the subject of this section.  The classification of the Lie algebras is constructive and it is exactly this construction through which lattices appear. In order to communicate the original work in this dissertation faster, we will omit proofs in this section. As recompense, we provide at least two references, at differing levels of formality, for the assertions made.
\subsubsection{To characterise a Lie algebra}
 An algebra has two components: its elements, and the operation which combines those two elements. Thus, a characterisation of an algebra must specify both what the algebra contains and all possible Lie brackets. In the sequel we are concerned to show that both are completely recoverable from the root system of the Lie algebra.
\subsubsection{The Cartan Weyl decomposition.}
A useful way to decompose a Lie algebra is by generalising the notion of `raising' and `lowering' operators. This is provided by the  \textit{Cartan-Weyl decomposition}. An essential object is the Cartan subalgebra (CSA henceforth):
\begin{definition}[Cartan subalgebra]
The Cartan subalgebra, $\mathcal{H}$, is the maximal set of generators $\{H_i\}$ such that:
\begin{equation}
[H_i,H_j]=0\text{   }\forall H_i,H_j\in \mathcal{H}
\end{equation}
\end{definition}
Introducing a basis we can say that, in a matrix lie algebra where the Lie bracket is just the matrix commutator, this would correspond to the set of maximally commuting matrices in that algebra. Hence, there is a set of simultaneous eigenvectors of all the elements of $\mathcal{H}$. 

To proceed further we introduce a particular representation, the \textit{adjoint representation}:

\begin{definition}
The adjoint representation, $ad$, is a representation which acts on the Lie algebra itself:
$$ad : \mathfrak{g}\rightarrow \mathfrak{g}$$
and with Lie bracket given by the operator commutator:
$$ad([S,T])=ad(S)\circ ad(T)-ad(T)\circ ad(S)\text{ }\forall S,T\in \mathfrak{l}$$
\end{definition}

Every Lie algebra has the adjoint representation as a representation. The action of the adjoint representation on an element of the Lie algebra is:
\begin{equation}
ad(S)T=[S,T]
\end{equation}
The adjoint representation allows us to introduce a symmetric bilinear form, the \textit{Killing form}, given as:
\begin{definition}[Killing form]
The Killing form $\Gamma$ is a symmetric, bilinear form on the Lie algebra:
$$\Gamma : \mathfrak{l}\times \mathfrak{l}\rightarrow \mathbb{R}$$
given by:
$$\Gamma(S,T)=\trace (ad(S)ad(T))$$
\end{definition}
This acts like a scalar product on the Lie algebra, and we can define: 
\begin{definition}
Let $\Gamma$ be the Killing form on the Lie algebra $\mathfrak{l}$. $\Gamma$ is said to be non-degenerate when for $S\in\mathfrak{l}$, $\Gamma(S,T)=0$  $\forall T\in\mathfrak{l}\implies S=0$.
\end{definition}
For which we have: \cite[p.~480,~C.10]{FH}, \cite[p.~22]{Humphreys} :
\begin{theorem}
$\Gamma$  is non-degenerate \textbf{iff} $\mathfrak{l}$ is semi-simple.
\end{theorem}
Since we are only concerned with semi-simple Lie algebras here, for us it is always non-degenerate. If we have a basis, $\Gamma$ takes the form $\gamma_{ij}$ and we can define $\gamma_{ij}\gamma^{jk}\req\delta_{i}^{k}$.

This now allows us to give one definition of a \textit{root}, which is crucial in making the connection to lattice theory. 
\begin{definition}[Root]
A root $\alpha$ is a vector whose components are the solutions to the equation:
$$ad(H_i)(T)=\alpha_iT$$
where $T\in \mathfrak{l}$.
\end{definition}
Since every element of $\mathcal{H}$ commutes, $T$ is an eigenvector of every element of $\mathcal{H}$. Thus, the $i^{\text{th}}$ component of $\alpha$ is the solution to the above eigenvalue equation with the $i^{\text{th}}$ element of $\mathcal{H}$. The number of linearly independent elements in $\mathcal{H}$ is called the \textit{rank} of the Cartan subalgebra and so the number of entries in $\alpha$ is equal to the rank of $\mathcal{H}$. The whole vector $\alpha$ is just the solution to the eigenvalue equation for a \textit{given} eigenvector of $ad(H)$, $T$ for every element of $\mathcal{H}$. The eigenvector $T$ is called a \textit{generator} of the Lie algebra. We will use $\Delta$ to denote the set of roots which we term the \textit{root space}. We can generalise the concept of a root to give the definition of a \textit{weight}:
\begin{definition}
Let $r : \mathfrak{l}\rightarrow GL(V,K)$ be a representation of the Lie algebra $\mathfrak{l}$, let $H\in \mathcal{H}$, and let $v\in V$ satisfy:
$$r(H)v=\omega(H)v$$
Then, the vector whose entries are the eigenvalues $\omega_i=\omega(H_i)$ is called the weight vector for that representation.
\end{definition}
\begin{remark}
The roots are the weights of the adjoint representation.
\end{remark}
Arbitrarily choosing a basis for the Lie algebra, we can introduce the notion of a \textit{positive weight}:
\begin{definition}[Positive weight]
Let $\{e_{i}\}$ be an arbitrarily ordered basis for the vector space of weights. Then the set of positive weights $\Omega$ consists of the elements $\omega$ such that the first non-zero component of $\omega\in\Omega>0$. A weight that is not positive is said to be negative.
\end{definition}

Although this partition is basis dependent, the results we will find from it are basis independent. In particular, this definition gives us a `direction' in weight space, which we will use to introduce the analogs of `raising' and `lowering' operators. \\
\indent Adopting the notation that $\mathfrak{l}_{\alpha}$ is the set of elements of the Lie algebra satisfying the eigenvalue equation with root $\alpha$, we state two facts:
\begin{fact}\label{commalphabeta}
$$[\mathfrak{l}_{\alpha},\mathfrak{l}_{\beta}]\subset \mathfrak{l_{\alpha+\beta}}$$
\end{fact}\cite[Proposition 7.17]{Hall2015}, or \cite[p.~58]{Lukas} 
%

\begin{fact}\label{eq: rootorth}
$T_{\beta}$ and $S_{\alpha}$ are orthogonal with respect to the scalar product $\Gamma$, unless $\alpha+\beta=0$. That is \cite[Appendix~D]{FH}, \cite[pg.58]{Lukas},:
\begin{equation}
\alpha+\beta\neq 0\implies\Gamma(T_{\alpha},S_{\beta})=0
\end{equation}
\end{fact}
%

We can then use the non-degeneracy of the Killing form to show \cite[Appendix~D]{FH}, \cite[p.~59]{Lukas}:
\begin{theorem}
If $\alpha$ is a root, then $-\alpha$ is also a root. 
\end{theorem}
 We want to associate $\alpha$ with a unique element of the Lie algebra $T$, and this is provided by (\cite[p.~93-95]{Georgi}, \cite[p.~59-65]{Lukas}, or  \cite[p.~173]{Hall2015}): 
\begin{theorem}\label{rootsubalgebra}
Every Lie algebra contains a sub-algebra corresponding to the generators and Cartan elements of the roots $\alpha$ and $-\alpha$. Each root corresponds to a unique generator and a unique CSA element. These elements are the \textit{Chevalley generators} defined as:
\begin{equation}
E_{\alpha}=\sqrt{\frac{2}{(\alpha,\alpha)}}X_{\alpha}
\end{equation}

\begin{equation}
F_{\alpha}=\sqrt{\frac{2}{(\alpha,\alpha)}}X_{-\alpha}
\end{equation}

\begin{equation}
H_{\alpha}=\frac{2}{(\alpha,\alpha)}\alpha^{i}H_{i},\text{  }\alpha^{i}\defeq\gamma^{ij}\alpha_{j}
\end{equation}
Where $X_{\alpha}$ is the unique generator corresponding to the root $\alpha$, i.e it is the only element of the Lie algebra satisfying $[H_i,X]=\alpha_i X_{\alpha}$.
\end{theorem}
This one-one correspondence between $\mathcal{H}$ and $\Delta$ us to define an inner form on root space:
\begin{equation}
(\alpha,\beta)\defeq\Gamma(H_{\alpha},H_{\beta})
\end{equation}
We thus have that every Lie algebra contains a subalgebra spanned by a pair of roots $\pm\alpha$. Crucially, the set of roots can reconstruct the \textit{entire} Lie algebra.
 That this is so was proved by Dynkin \cite{Dynkin}, although we will later use the form due to Serre \cite[Chapter~6]{Serre} and our discussion largerly follows \cite{West}.
We can now define:
\begin{property}[Highest weight]
Every irreducible representation $r: \mathfrak{l}\rightarrow GL(V,K)$ has a \textbf{unique} \textit{highest weight vector} $\lambda\in V$ defined as:
$$E_{\alpha}\lambda=0\text{   }\forall \alpha\in\Delta_S$$
\end{property}
\begin{definition}[Positive root]
Arbitrarily choose some root $\beta$. The \textit{positive roots} $\alpha_+$ are those for which $\Gamma(H_{\alpha},H_{\beta})>0$. We denote the set of positive roots as $\Delta_{+}$.
\end{definition}
A basis for root space is given by the \textit{simple roots}:
\begin{definition}[Simple roots]
A simple root is a positive root that cannot be expressed as the sum of two positive roots.
\end{definition}
Given a partition into positive and negative roots, there is a \textit{unique} set of simple roots, $\Delta_S$, which also form a basis of $\Delta$.
\begin{definition}[Cartan matrix]
The \textit{Cartan matrix}, by convention $A$, is the matrix whose entries consists of the scalar products on the simple roots\footnote{We use the convention in \cite{FH} for the ordering of the indices $\alpha,\beta$. Note that \cite{West} uses the reverse convention.}
\begin{equation}
A_{\alpha\beta}=\frac{2(\alpha,\beta)}{(\beta,\beta)}; {\alpha,\beta}\in\Delta_S
\end{equation}
\end{definition}
%
%

The Cartan matrix is derived solely from the Killing form on the simple roots. Remarkably, the Cartan matrix alone is sufficient to  reconstuct the entire Lie algebra. The most elegant statement of this fact is \cite[Chapter~6]{Serre}:
\vspace{-1mm}
\begin{theorem}[Serre's theorem]
Let $\mathfrak{l}$ be a finite dimensional semi-simple Lie algebra. The five conditions:
\begin{equation}\label{eq:Serre1}
[H_{\alpha},H_{\beta}]=0
\end{equation}
\vspace{-12mm}
\begin{equation}\label{eq:Serre2}
[E_{\alpha},F_{\alpha}]=\delta_{\alpha\beta}H_{\alpha}
\end{equation}
\vspace{-12mm}
\begin{equation}\label{eq:Serre3}
\begin{aligned}
& [H_{\alpha},E_{\beta}]=A_{\beta\alpha}E_{\beta}\\
& [H_{\alpha},F_{\beta}]=-A_{\beta\alpha}F_{\beta}
\end{aligned}
\end{equation}
\vspace{-5mm}
\begin{equation}\label{eq:Serre4}
\begin{aligned}
& ad(E_{\alpha})^{1-A_{\beta\alpha}}E_{\beta}=0\text{ or :}\\
& [E_{\alpha},[E_{\alpha},[E_{\alpha}....,E_{\beta}]]\hdots]=0
\end{aligned}
\end{equation}
\vspace{-5mm}
\begin{equation}\label{eq:Serre5}
\begin{aligned}
& ad(F_{\alpha})^{1-A_{\beta\alpha}}F_{\beta}=0\text{ or :}\\
& [F_{\alpha},[F_{\alpha},[F_{\alpha},...,F_{\beta}]]...]=0
\end{aligned}
\end{equation}
uniquely reconstruct $\mathfrak{l}$.
\end{theorem}
This is a map from $\Delta_S$ to $\mathfrak{l}$. Every  $\{\alpha,-\alpha\}\in\Delta_S$ is, by Theorem \ref{rootsubalgebra}, associated to a unique set of three Chevalley generators $\{E_{\alpha},F_{\alpha},H_{\alpha}\}$ that satisfy the appropriate eigenvalue relation. We will refer to this map as $\mathcal{S}: \Delta_S\rightarrow\mathfrak{l}$. We note in passing that the first three relations apply even to infinite-dimensional Lie algebras \cite{KacEDD}. \\
\indent From Fact \ref{commalphabeta} the commutator $[E_{\alpha},E_{\gamma}]\subset\mathfrak{l}_{\alpha+\gamma}$ and since the $\mathfrak{l}_{\alpha+\gamma}$ subalgebra is spanned by a single generator\footnote{This follows from the fact that we can associate a unique generator to every root and that $\mathfrak{l}_{\alpha}$ is defined as the set of elements satisfying the eigenvalue equation for $\alpha$ and theorem $3.23$ guarantees that these can only be scalar multiples of the generator $E_{\alpha}$} it corresponds to the generator associated to the root $\alpha+\gamma$. Taking all possibe commutators, we associate to every root a space $\mathfrak{\alpha}$ in the Lie algebra. Since every Lie algebra gives rise to a root system and Serre's theorem guarantees that the Chevalley generators reconstruct a unique Lie algebra, it follows that a classification of the root systems classifies the Lie algebras.
\subsection{Summary}\hfill

	\indent Let us summarise the results of this subsection. A Lie algebra gives rise to a set of roots $\Delta$, spanned by linear combinations of the simple roots $\Delta_S\subset\Delta_{+}$. The inner form on the simple roots gives rise to a Cartan matrix and the Cartan matrix then determines the Serre relations. The roots are able to reconstruct the Lie algebra, and so studying the roots give us a simple way to study the Lie algebras. Just like before, we summarise this with the commutative diagram:

\usetikzlibrary{matrix}
\begin{figure}[h]

\begin{tikzpicture}
  \matrix (m) [matrix of math nodes,row sep=3em,column sep=4em,minimum width=2em] {
     \mathfrak{l} & \mathfrak{\tilde{l}} \\
     \Phi & \tilde{\Phi} \\};
  \path[-stealth]
    (m-2-1) edge node [left] {$\mathcal{S}$} (m-1-1)
    (m-1-1) edge node [above] {$\mathcal{F}$} (m-1-2)       
    (m-2-1) edge node [above] {$\mathcal{M}$}  (m-2-2)
    (m-2-2) edge node [right] {$\tilde{\mathcal{S}}$} (m-1-2);
     
\end{tikzpicture}
\caption{Root system-Lie algebra Correspondence} \label{fig: Fig.2}
\end{figure}
Thus, to classify the Lie algebra's, we just need to classify the Cartan matrices arising from their simple roots. 
\subsection{Classification of Cartan Matrices}\hfill

To do so we follow \cite{West}. Here we provide proofs, since we will later  see that the classification of Cartan matrices is also a classification of lattices. The crucial lemma is:
\begin{lemma}\label{lem:HWV}
Let $\ket{\lambda}$ be the highest weight vector of an irreducible representation (labelled by the corresponding highest weight $\lambda$) of a finite dimensional Lie algebra $\mathfrak{l}$, and let $\omega$ denote a weight in that algebra. Then:
\begin{equation}
a_{\alpha}\defeq\frac{2(\lambda,\alpha)}{\alpha,\alpha}\in\mathbb{Z}
\end{equation}
\end{lemma}
\begin{strategy}
Our strategy is to set up a recurrence relation on the normalisation constant of the weights acted on by $E_{\alpha}$, which contains information about the lowest weight vector.
\end{strategy} 
\begin{proof}
We have that:
\begin{equation}
(F_{\alpha})^{n}\ket{\lambda}=m_{n}\ket{\lambda-n\alpha}
\end{equation}
where $m_{n}$ is a normalisation constant. Since we may choose a normalisation for each of the states \textit{once}, we choose it such that $r_{n}=1$, for the states created when the highest weight vector $\ket{\lambda}$ is acted on by the `lowering operator' $F_{\alpha}$. Note that once we choose to do this we have to live with the fact that the action of the `raising operator' $E_{\alpha}$ on the states does not, in general, lead to a normalised vector. Then consider:
\begin{equation}\label{eq:normWest}
\begin{aligned}
& E_{\alpha}\ket{\lambda-n\alpha}=r_{n}\ket{\lambda-(n-1)\alpha}=E_{\alpha}F_{\alpha}\ket{\lambda-(n-1)\alpha}\\
& =([E_{\alpha},F_{\alpha}]+F_{\alpha}E_{\alpha})\ket{\lambda-(n-1)\alpha}\\
& =(H_{\alpha}+r_{n-1})\ket{\lambda-(n-1)\alpha}\text{ (using Serre relation }\ref{eq:Serre2})
\end{aligned}
\end{equation}.
From the definition of a weight we have that $r(H_{\alpha})v=\omega(H_{\alpha})v$ and we can choose a basis for a finite dimensional representation $r$ such that the weights are also eigenstates of $\mathcal{H}$ so that $H_i\ket{\omega}=\omega_i\ket{\omega}$. Then using the definition of the Chevalley generator $H_{\alpha}=\frac{2}{(\alpha,\alpha)}\alpha^iH_i$, we have that:
\begin{equation}
H_{\alpha}\ket{\omega}=\frac{2(\omega,\alpha)}{(\alpha,\alpha)}\ket{\omega}
\end{equation}
Since $(\omega,\alpha)=\Gamma({\omega,\alpha})=\sum\limits_{ij}\alpha_{i}\gamma^{ij}\omega_j=\alpha^{j}\omega_j$. Substituting into equation \ref{eq:normWest}:
\begin{equation}
\begin{aligned}
& r_{n}\ket{\lambda-(n-1)\alpha}=\bigg(\frac{2(\alpha,(\lambda-(n-1)\alpha))}{(\alpha,\alpha)}+r_{n-1}\bigg)\ket{\lambda-(n-1)\alpha}\text{ or :}\\
& r_{n}=r_{n-1}+\frac{2(\alpha,\lambda)}{(\alpha,\alpha)}-2(n-1)
\end{aligned}
\end{equation}
Since we must have $E_{\alpha}$ eliminate the highest weight vector, $r_0=0$. So $r_1=\frac{2(\alpha,\lambda)}{(\alpha,\alpha)}$ and $r_n$ will have $n$ such factors. $r_2=2\big(\frac{2(\alpha,\lambda)}{(\alpha,\alpha)}\big)-2$ and we see that the final term leads to an arithmetic progression:
\begin{equation}
2\sum_{i=1}^{n-1}i=n(n-1)
\end{equation}
and so:
\begin{equation}
r_{n}=2n\frac{(\alpha,\lambda)}{(\alpha,\alpha)}-n(n-1)
\end{equation}
Since the representation is finite dimensional, there can only be a finite number of weights, and hence the action of the lowering operator $F_{\alpha}$ must eliminate the state. Alternately put, there is some \textit{positive integer} $a_{\alpha}$ such that:
$$
\begin{aligned}
& F_{\alpha}^{a_\alpha+1}=0
& \implies r_{a_{\alpha}}=(a_{\alpha}+1)\big(2\frac{(\alpha,\lambda)}{(\alpha,\alpha)}-a_{\alpha}\big)\req0
& \implies a_{\alpha}=2\frac{(\alpha,\lambda)}{(\alpha,\alpha)}
\end{aligned}
$$
The set of positive integers $a_{\alpha}$ with $\alpha\in\Delta_S$ is the well known \textit{Dynkin label} of the highest weight of an irrep. Moreover, we have shown that it is an integer, as required.
\end{proof}
Lemma \ref{lem:HWV} can straightforwardly be extended from the highest weight $\lambda$ to hold for \textit{all} weights $\omega$:
\begin{theorem}
Let $r$ be a finite dimensional irreducible representation of a finite dimensional Lie algebra $\mathfrak{l}$ with a highest weight vector $\lambda$. Then let $\omega$ be an arbitrary weight with weight vector $\ket{\omega}$. Then:
\begin{equation}
2\frac{(\omega,\alpha)}{(\alpha,\alpha)}=q-p\in\mathbb{Z}
\end{equation}
\end{theorem}
\begin{proof}
There is some $p$ such that:
\begin{equation}
(E_{\alpha})^{p}\ket{\omega}=0
\end{equation}
and some $q$ such that:
\begin{equation}
(F_{\alpha})^{q}\ket{\omega}=0
\end{equation}
The first equation simply asserts the existence of a highest weight vector, reached by $p$ actions of the raising operator on $\ket{\omega}$. The second reflects the fact that we are dealing with a finite dimensional representation and so there must be a finite number of weights. The weight corresponding to $(E_{\alpha})^{n}\ket{\omega}$ is $\omega+n\alpha$. The sequence of weights formed in this way, by actions of raising and lowering operator of an $A_1$ subalagbera;
\begin{equation}
\lambda=\omega+p\alpha,\omega+(p-1)\alpha,...,\omega,\omega-\alpha,...,\omega-q\alpha=\lambda-a_{\alpha}
\end{equation}
is termed the \textit{$\alpha$ sequence through $\omega$}. It is a sequence of all the weights that can be formed by adding integer multiples of the root $\alpha$ to the weight $\omega$. Note that in the last equality we have used the fact shown in the proof of Lemma \ref{lem:HWV} that the number of lowering operators that must be applied to reach the lowest weight is equal to the Dynkin label $a_{\alpha}$. \\
We have:
\begin{equation}
a_\A=\frac{2(\lambda,\alpha)}{(\A,\A)}=\frac{2(\omega+p\A,\A)}{(\A,\A)}=\frac{2(\omega,\A)}{(\A,\A)}+2p\implies a_\A-2p=\frac{2(\omega,\A)}{(\A,\A)}
\end{equation}
Then, since $\lambda=\omega+p\A$ and $\omega-q\A=\lambda-a_\A$ we have that:
\begin{equation}
a_\A\alpha=\lambda-(\omega-q\A)=\lambda-(\lambda-p\A)+q\A
\end{equation}
\begin{equation}
\implies a_{\A}=q+p
\end{equation}
But since $p,q\in\mathbb{Z}$ this entails that:
\begin{equation}\label{eq:weightIP}
\frac{2(\omega,\alpha)}{(\A,\A)}=q-p\in\mathbb{Z}
\end{equation}
thus proving the theorem. 
\end{proof}
\begin{corollary}\label{cor:simplethings}
Since the adjoint representation is always a finite dimensional representation of a finite dimensional Lie algebra, and since the weights of the adjoint representation are the roots, we have that:
\begin{equation}
2\frac{(\alpha,\beta)}{(\A,\A)}=q-p\in\mathbb{Z}\text{  }\forall\A,\beta\in\Delta
\end{equation}
\end{corollary}
\begin{corollary}\leavevmode
\begin{enumerate}
	\item If, $\beta\neq\pm\alpha$ then $(\beta,\alpha)\leq 0\implies(\alpha+\beta)\in\Delta$. \\
	\item If, $\beta\neq\pm\alpha$ then $(\beta,\alpha)\geq 0\implies(\alpha-\beta)\in\Delta$
\end{enumerate}
\end{corollary}

\begin{proof}\leavevmode
\begin{enumerate}
	\item This follows since $2\frac{(\alpha,\beta)}{\A,\A}=q-p$, where the RHS is the number of times the weight is \textit{lowered} minus the number of times it is \textit{raised}. If the antecedent of $(1)$ is satisfied, then the weight is raised more times than it is lowered, and the $\alpha$ string through $\beta$ must contain $(\alpha+\beta)$. \\
	\item Similarly, if the antecedent of $(2)$ is satisfied, then the weight is raised more times than it is lowered so the $\alpha$ string through $\beta$ must contain $(\alpha-\beta)$.
\end{enumerate}
\end{proof}
\begin{lemma}\label{lem:simplethings2}
If $\alpha\neq\beta$ and $\alpha,\beta\in\Delta_S$, then $\alpha-\beta$ is not a root.
\end{lemma}
\begin{proof}
Assume for contradiction that this is not so. Then, either $\alpha-\beta$ or $\beta-\alpha$ is a positive root. In the first case: $\beta+(\alpha-\beta)=\alpha$, and in the second: $\A+(\beta-\A)=\beta$. Both cases contradict the simple property of $\alpha$ or $\beta$ respectively. 
\end{proof}
\begin{corollary}
$$(\alpha,\beta)\in\Delta_S,\text{ }\A\neq\pm\beta\implies(\A,\beta)\leq 0$$
\end{corollary}
\begin{proof}
We prove the contrapositive, that $(\alpha,\beta)>0\implies(\alpha,\beta)\not\in\Delta_S$ for $\alpha\neq\pm\beta$. Assume that $(\alpha,\beta)>0$. Then by $(2)$ of Corollary \ref{cor:simplethings}, $(\alpha-\beta)$ is a root, contradicting Lemma \ref{lem:simplethings2}. So $(\alpha-\beta)\not\in\Delta_S$, as required.
\end{proof}
The consequence of the above is that the values of the Cartan matrix are very strongly constrained. For this we need the lemma \cite[Chapter~14,~14.22]{FH}:

\begin{lemma}\label{Gammagood}
$\Gamma_{\restriction(\mathcal{H}\times\mathcal{H})}$ is a symmetric, bilinear, positive definite inner form on $\mathcal{H}\times\mathcal{H}$ over the field $\mathbb{R}$.
\end{lemma}
This allows us to identify the inner product on roots space as the standard inner product on Euclidean space, and so:
\begin{equation}
\frac{(\A,\beta)}{\sqrt{(\A,A)(\beta,\beta)}}=\cos(\theta)
\end{equation}
\begin{equation}
\implies A_{\alpha\beta}A_{\beta\A}=\frac{2(\A,\beta)}{(\beta,\beta)}\frac{2(\beta,\A)}{(\A,\A)}=4\cos^2(\theta)
\end{equation}
By Lemma \ref{Gammagood}, the LHS can only take integer values, and thus there are only six possibilities in which $\A\neq\pm\beta$ (the case $\cos(\theta)=1$) and  $\norm{\alpha} \defeq +\sqrt{(\alpha,\alpha)}$:\footnote{Note that $A_{\alpha\beta}$ and $A_{\beta\alpha}$ must have the same sign since $(\alpha,\beta)=(\beta,\alpha)$.}


\begin{table}[h]

\begin{tabular}{SSSS} \toprule
    {$A_{\alpha\beta}$} & {$A_{\beta\A}$} & {$\cos(\theta)$} & {$\theta$}\\ \midrule
    -3  & -1 & $\minus\frac{\sqrt{3}}{2}$ & $\frac{5\pi}{6}$ \\
    -2  & -1 & $\minus\frac{\sqrt{2}}{2}$ & $\frac{3\pi}{4}$   \\
    -1  & -1 & $\minus\frac{1}{2}$ & $\frac{2\pi}{3}$  \\ \midrule
    0 & 0&0& $\frac{\pi}{2}$    \\ \midrule
    1  & 1  & $\frac{1}{2}$ & $\frac{\pi}{6}$    \\ 
    2  & 1 & $\frac{\sqrt{2}}{2}$ & $\frac{\pi}{4}$   \\
    3  & 1 & $\frac{\sqrt{3}}{2}$ &  $\frac{\pi}{3}$   \\ \bottomrule
\end{tabular}
\\
\caption{A table of the possible values of the `angle' $\theta$ between the roots, assuming that $\norm{\alpha}\geq \norm{\beta}$, where $\norm{\alpha} \defeq +\sqrt{(\alpha,\alpha)}$}
\label{Dynkinvalues}
\end{table}
Since any \la corresponds to a unique Cartan matrix and the above table lists all possible Cartan matrices, this completes our stated aim of classifying the Lie algebras. The simple roots alone are sufficient to reconstruct the Lie algebra, so we can consider only the first four rows of Table \ref{Dynkinvalues}. Dynkin introduced \cite{Dynkin} an elegant notation for the Cartan matrix. The roots with larger norm are displayed as empty circles, and the `shorter' roots are displayed as filled circles. Connections between roots denote the angles between them. \\
\begin{figure}[h]
\begin{center}
  \begin{tikzpicture}[scale=.4]
    \draw (-1,0) node[anchor=east]  {$A_n$};
    \foreach \x in {0,...,5}
    \draw[xshift=\x cm,thick] (\x cm,0) circle (.3cm);
    \draw[dotted,thick] (0.3 cm,0) -- +(1.4 cm,0);
    \foreach \y in {1.15,...,4.15}
    \draw[xshift=\y cm,thick] (\y cm,0) -- +(1.4 cm,0);
  \end{tikzpicture}
\end{center}

\begin{center}
  \begin{tikzpicture}[scale=.4]
    \draw (-1,0) node[anchor=east]  {$B_n$};
    \foreach \x in {0,...,4}
    \draw[xshift=\x cm,thick] (\x cm,0) circle (.3cm);
    \draw[xshift=5 cm,thick,fill=black] (5 cm, 0) circle (.3 cm);
    \draw[dotted,thick] (0.3 cm,0) -- +(1.4 cm,0);
    \foreach \y in {1.15,...,3.15}
    \draw[xshift=\y cm,thick] (\y cm,0) -- +(1.4 cm,0);
    \draw[thick] (8.3 cm, .1 cm) -- +(1.4 cm,0);
    \draw[thick] (8.3 cm, -.1 cm) -- +(1.4 cm,0);
  \end{tikzpicture}
\end{center}

\begin{center}
  \begin{tikzpicture}[scale=.4]
    \draw (-1,0) node[anchor=east]  {$C_n$};
    \foreach \x in {0,...,4}
    \draw[xshift=\x cm,thick,fill=black] (\x cm,0) circle (.3cm);
    \draw[xshift=5 cm,thick] (5 cm, 0) circle (.3 cm);
    \draw[dotted,thick] (0.3 cm,0) -- +(1.4 cm,0);
    \foreach \y in {1.15,...,3.15}
    \draw[xshift=\y cm,thick] (\y cm,0) -- +(1.4 cm,0);
    \draw[thick] (8.3 cm, .1 cm) -- +(1.4 cm,0);
    \draw[thick] (8.3 cm, -.1 cm) -- +(1.4 cm,0);
  \end{tikzpicture}
\end{center}

\begin{center}
  \begin{tikzpicture}[scale=.4]
    \draw (-1,0) node[anchor=east]  {$D_n$};
    \foreach \x in {0,...,4}
    \draw[xshift=\x cm,thick] (\x cm,0) circle (.3cm);
    \draw[xshift=8 cm,thick] (30: 17 mm) circle (.3cm);
    \draw[xshift=8 cm,thick] (-30: 17 mm) circle (.3cm);
    \draw[dotted,thick] (0.3 cm,0) -- +(1.4 cm,0);
    \foreach \y in {1.15,...,3.15}
    \draw[xshift=\y cm,thick] (\y cm,0) -- +(1.4 cm,0);
    \draw[xshift=8 cm,thick] (30: 3 mm) -- (30: 14 mm);
    \draw[xshift=8 cm,thick] (-30: 3 mm) -- (-30: 14 mm);
  \end{tikzpicture}
\end{center}

\begin{center}
  \begin{tikzpicture}[scale=.4]
    \draw (-1,0) node[anchor=east]  {$G_2$};
    \draw[thick] (0,0) circle (.3 cm);
    \draw[thick,fill=black] (2 cm,0) circle (.3 cm);
    \draw[thick] (30: 3mm) -- +(1.5 cm, 0);
    \draw[thick] (0: 3 mm) -- +(1.4 cm, 0);
    \draw[thick] (-30: 3 mm) -- +(1.5 cm, 0);
  \end{tikzpicture}
\end{center}

\begin{center}
  \begin{tikzpicture}[scale=.4]
    \draw (-3,0) node[anchor=east]  {$F_4$};
    \draw[thick] (-2 cm ,0) circle (.3 cm);
    \draw[thick] (0,0) circle (.3 cm);
    \draw[thick,fill=black] (2 cm,0) circle (.3 cm);
    \draw[thick,fill=black] (4 cm,0) circle (.3 cm);
    \draw[thick] (15: 3mm) -- +(1.5 cm, 0);
    \draw[xshift=-2 cm,thick] (0: 3 mm) -- +(1.4 cm, 0);
    \draw[thick] (-15: 3 mm) -- +(1.5 cm, 0);
    \draw[xshift=2 cm,thick] (0: 3 mm) -- +(1.4 cm, 0);
  \end{tikzpicture}
\end{center}

\begin{center}
  \begin{tikzpicture}[scale=.4]
    \draw (-1,1) node[anchor=east]  {$E_6$};
    \foreach \x in {0,...,4}
    \draw[thick,xshift=\x cm] (\x cm,0) circle (3 mm);
    \foreach \y in {0,...,3}
    \draw[thick,xshift=\y cm] (\y cm,0) ++(.3 cm, 0) -- +(14 mm,0);
    \draw[thick] (4 cm,2 cm) circle (3 mm);
    \draw[thick] (4 cm, 3mm) -- +(0, 1.4 cm);
  \end{tikzpicture}
\end{center}

\begin{center}
  \begin{tikzpicture}[scale=.4]
    \draw (-1,1) node[anchor=east]  {$E_7$};
    \foreach \x in {0,...,5}
    \draw[thick,xshift=\x cm] (\x cm,0) circle (3 mm);
    \foreach \y in {0,...,4}
    \draw[thick,xshift=\y cm] (\y cm,0) ++(.3 cm, 0) -- +(14 mm,0);
    \draw[thick] (4 cm,2 cm) circle (3 mm);
    \draw[thick] (4 cm, 3mm) -- +(0, 1.4 cm);
  \end{tikzpicture}
\end{center}

\begin{center}
  \begin{tikzpicture}[scale=.4]
    \draw (-1,1) node[anchor=east]  {$E_8$};
    \foreach \x in {0,...,6}
    \draw[thick,xshift=\x cm] (\x cm,0) circle (3 mm);
    \foreach \y in {0,...,5}
    \draw[thick,xshift=\y cm] (\y cm,0) ++(.3 cm, 0) -- +(14 mm,0);
    \draw[thick] (4 cm,2 cm) circle (3 mm);
    \draw[thick] (4 cm, 3mm) -- +(0, 1.4 cm);
  \end{tikzpicture}
\end{center}
\caption{The classification of all finite dimensional semi-simple Lie algebras through Dynkin diagrams}
\label{DynkinDiagrams}
\end{figure}

\subsubsection{Review}
Let us review what we have shown in this section. We introduced the Cartan decomposition of a Lie algebra and, employing the adjoint representation, defined a symmetric, positive definite bilinear form, the Killing form, on the Lie algebra. Using the Killing form, we showed that every Lie algebra contains $A_1$ as a subalgebra (Theorem \ref{rootsubalgebra}), generated by the Chevalley generators.  The Serre construction showed that every Lie algebra could be uniquely reconstructed from the Cartan matrix consisting of inner products of simple roots. Thus, we can study Lie algebras from their root systems in an exactly analogous way to studying a Lie group by studying it's Lie algebra.\footnote{The motivation is even stronger in the latter since we recover \textit{all} of the algebra.} Classifying the Cartan matrices, then, classifies all the possible finite dimensional semi-simple Lie algebra.\\
\subsection{Lattices}\hfill\label{latticepp}

What, then, is the simplest way we can go about systematically studying the root systems of Lie algebras? Collecting our results and noting that equation \ref{eq:weightIP} implies:
\begin{equation}
W_{\alpha}(\omega)\defeq\omega-\frac{2(\omega,\A)}{(\A,\A)}\alpha=\omega-(q-p)\A=\lambda-q\A\in\Delta
\end{equation}
we can define a \textit{root system} independently of the Lie algebra:

\begin{definition}[Root system]
Let $V$ be a finite dimensional vector space. Then the root system $\Delta(V,R)$ is a finite set of non-zero vectors (the ``roots") with the standard Euclidean inner form $(\cdot,\cdot)$ that satisfy:\\
\begin{enumerate}
	\item The roots span $V$.
	\item The only scalar multiples of a root $\alpha$ belonging to $\Delta$ are $\alpha$ and $-\alpha$. 
	\item If $\alpha,\beta\in\Delta$, then $W_{\alpha}(\omega)\defeq\omega-\frac{2(\omega,\A)}{(\A,\A)}\alpha\in\Delta$ 
	\item $\alpha,\beta\in\Delta\implies \frac{2(\alpha,\beta)}{(\beta,\beta)}\in\mathbb{Z}$ 
\end{enumerate}
\end{definition}
Note that $s_{\alpha}(\alpha)=-\alpha$, so reflection of a root through the hyperplane orthogonal to it yields it's negative, which is another root. Previously the only reason for the roots to be priviliged in an even lattice was because they where the minimal possible length vectors - now we see that roots also generate automorphisms (in particular, reflections) of the lattice. 

Choosing the $\mathbb{Z}\text{-submodule}$ of $\Delta$ generates a \textit{lattice}, with an inner form that is precisely the Cartan matrix of the corresponding Lie algebra. 

\begin{definition}[root]
A \textit{root} of a lattice $L$ is some $\alpha\in L$ such that the mapping $s_{\alpha} : L\otimes\mathbb{R}$ which acts on $l\in L$ as:
\begin{equation}
s_{\alpha} : l\rightarrow l-\frac{2(l,\alpha)}{(\alpha,\alpha)}
\end{equation}
maps $L$ to itself - that is, the map $s_{\alpha}$ is a reflection isometry of the $L$.
\end{definition}

\subsection{Modelling}\hfill

We take a moment to consider how lattices can be used to construct theories which contain the standard model embedded in them.The classical concerns for modellers looking to extend the standard model are:
\begin{enumerate}
\item Which Lie algebras contain the standard model as a subalgebra?
\item Which irreducible representations are contained within the wider theory that is not in the standard model?
\item What is the phenomenology of these matter representations?
\end{enumerate}
\indent With what we have already we can use lattice theory to answer the first two questions - often in a systematic way. Consider two root lattices $R_1$, $R_2$. To what Lie algebra does $R_1\oplus R_2$ correspond? The Serre construction requires that it is the Lie algebra defined by the corresponding Cartan matrix of $R_1\oplus R_2$. Since we take the orthogonal sum, this is just the block diagonal matrix with blocks $A_1$ and $A_2$\footnote{Where $A_i$ is the Cartan matrix of $R_{i}$.}. This in turn corresponds to a Lie algebra that is the Cartesian product of the Lie algebra $\mathfrak{l}_1$ and $\mathfrak{l}_2$\footnote{The notation is the same as that of the Cartan matrices.}.
\begin{fact}
For Lie algebras $\mathfrak{l}_1,\hdots,\mathfrak{l}_n$ with root lattices $R_1,\hdots,R_n$, the direct sum of the root lattices is the root system of a Lie algebra that is the cartesian product of the $\mathfrak{l}_i$:
\begin{equation}
\bigoplus\limits_{i=1}^{n}R_i\estimates\bigtimes\limits_{i=1}^{n}\mathfrak{l}_i
\end{equation}
\end{fact}
We can reverse the argument, and observe that if a Cartan matrix $A$ contains a submatrix that is equivalent to a Cartan matrix $A'$, then the corresponding Lie algebra's satisfy $\mathfrak{l}'\subseteq\mathfrak{l}$.\\
\indent We can extend the Dynkin classification to give the \textit{extended Dynkin diagram} \cite[p.~53-55]{KacEDD} formulism which provides a straightforward way to see which subalgebras are contained within another, thus answering the first point. 

\subsection{Summary}\hfill

 Thus we have seen that a root lattice is a simple structure that we can use to study a Lie algebra and hence a Lie group. Given the central role that Lie algebras and their classification have played in theoretical physics, we can hope to gain physical insight if we are able to translate a physical problem into lattice terms. The rest of this dissertation is an attempt to do just that. 
%
\section{Physical motivation}\label{Ftheory}
We are now in a position to state the specific physical motivation of this dissertation. Since our concern is primarily mathematical, the following is a very coarse outline of the physical intepretation of the conjecture we will later give. We omit formal definitions, as to introduce the objects discussed properly would take us too far afield. A thorough review is given in \cite{1401}. 

\subsection{F-theory and the string theory landscape.}\hfill


A central concern for string theory is to find the low-energy effective action, and hence the phenomenology, corresponding to different string theoretic models. F-theory is a non-pertubative completion of Type IIB string theory and is a promising candidate for string phenomenology \cite{Fpheno1}, \cite{Fpheno2}.\\
\indent What is an F-theory model? Like other string theories F-theory can be considered as a limit of M-theory and so will contain equations of motion in eleven dimensions. To specify a model we must compactify the remaining dimensions, and to do so we require three things:
\begin{enumerate}
	\item A complex four-manifold $X_4$. The standard choice is a \textit{Calabi-Yau four-fold} $CY_4$. 
	\item An \textit{elliptic fibration} over the complex four-manifold: $X_{4}\rightarrow B3$. Where $B3$ is a complex three dimensional space.
	\item A quantised \textit{four-flux}  $G_4$. This can be thought of as a generalised `vector potential'. Introducing a basis, it will be a rank 4 anti-symmetric tensor which appears in the equations of motion.
\end{enumerate}
\subsection{Where is the physics?}\hfill

The physics of this model is contained within it's geometry. There is a particular class of ``well-behaved" deformations of the manifold - these are it's moduli and their space is the \textit{moduli space}. There are two types of deformations of a $K3$ surface: the \textit{K\"ahler} moduli and the \textit{complex structure moduli}. We are concerned only with the complex structure moduli. These objects can be interpreted as the \textit{fields} of a four-dimensional effective theory. The equations of motion dictate the structure of the manifold and the four-flux, $G_4$, enters these equations the higher dimensional theory. This causes the moduli, $\phi_i$ to take on certain values. The $\phi_i$ and the $G_4$ then define an effective potential in four-dimensions $V_{\textrm{eff}}(\phi_i,G_4)$. \\
\indent A $K3$ surface $K$ defines a lattice, the \textit{Neron-Severi lattice}, $S_K$. The rank of this lattice is the \textit{Picard number}. An elliptic fibration is given by specifying a primitive embedding of the hyperbolic lattice $U$ into $S_K$ \cite{Nishiyama}. We make the following choices:
\begin{enumerate}
\item $X_4=K3\times K3$.
\item The elliptic fibration is only non-trivial on one of the $K3$. That is, the elliptic fibrations we consider are of the form: $K3\times K3\rightarrow K3\times\mathbb{CP}^{1}$.
\item We consider \textit{extremal} $K3$ \textit{surfaces} which have $\rho_K=20$. These were introduced in \cite{Dasgupta:1999ss}. This choice is made in order that all of the complex structure moduli are fixed - which occurs for a family of $G_4$. A member of this family occurs as the four-flux in our models.
\end{enumerate}
\indent The manifold will in general contain \textit{Du Val} singularities. These singularities can be classified by the $ADE$ classification. The type of singularity gives rise to a gauge group of the effective theory \cite{ADEsingAsp}, \cite{1401}.
\subsection{Where are the lattices?}\hfill

Lattice theory is used here to study the \textit{homology group} of the singularities once they undergo \textit{resolution} to make them smooth. For our purpose, the homology group can be thought of as the group of submanifolds that both have no boundary and are not themselves a boundary. Technically, for a space $\mathcal{M}$, the $d$-th homology is:
\begin{equation}
\frac{\ker(\partial^{d-1})}{\mathrm{Im}(\partial^{d})}\defeq H^{d}(\mathcal{M})
\end{equation}
Taking the space of such sub-manifolds over $\mathbb{Z}$ forms a lattice - with each lattice point corresponding to an equivalence class of manifolds. In our case the manifold is $\frac{\mathbb{C}^2}{\Gamma}$, where $\Gamma\subseteq SU(2)$. This manifold has a singularity at the origin and so we consider the second homology of the resolution of $\frac{\mathbb{C}^2}{\Gamma}$. The positive roots of the resulting lattice correspond to $\mathbb{CP}^{1}$ surfaces in the resolved manifold.\\
\indent For a torus, the first homology, $H^1$, corresponds to the signature $(1,1)$ \textit{hyperbolic plane lattice} $U$ with Gramian:
\begin{equation}
G^{U}\defeq\begin{pmatrix}0&1\\1&0\end{pmatrix}
\end{equation}
The second homology of a $K3$ surface $X$, $H_2(X,\mathbb{Z})$, forms a signature $(3,19)$ lattice isometric to $L^{(3,19)}\defeq U^{\oplus 3}\oplus (E_{8})^{\oplus 2}$.
\subsubsection{Analysing the model}\footnote{In this discussion and in the discussion of the Kneser-Nishiyama method we will use a signature that differs by a sign to that which we will use in the rest of the dissertation. This is done for the convenience of the reader who wishes to obtain further detail in the references, where this is standard.}
Once the Calabi-Yau four-fold, the elliptic fibration and the four-form flux are given, the model is completely specified. Now that we have a class of models, we would like to analyse them. In particular, we want to ask the question:
\begin{qst}
Do all such F-theory models contain non-abelian gauge groups? That is, do non-abelian gauge groups occur generically in such models?
\end{qst}
To answer this question we need to classify the possible elliptic fibrations over $K3$. In \cite{1401} it was shown that the answer to this question is contained in a \textit{frame lattice} $W$. The theorem \cite{1401}: 
\begin{theorem}
The root system $W_{\mathrm{root}}$ of the frame lattice $W$ is a sum of root lattices which generate the non-abelian gauge groups: 
$$W_{\mathrm{root}}=\bigoplus_{i}R_i$$
\end{theorem}
entails that the frame lattice $W$ contains the answer to this question, and $W$ will be the ultimate object of our study.
In particular, if $W_{\mathrm{root}}=\emptyset$ then there are no non-abelian gauge groups in the low-energy effective action, and so the equivalent question in lattice theoretic terms is:
\begin{qst}
Do all frame lattices $W_{\mathrm{frame}}$ corresponding to an elliptic fibration over a $K3$ surface contain roots?
\end{qst}
Crucially, there is not an infinite number of possible lattices $W$. If there was, then the answer would be no. For the F-theory model to be anomaly free the \textit{tadpole cancellation condition} for $D3$ branes, first introduced in \cite{SETHI1996213}, with the connection to the current setting spelled out in \cite[pg.~7]{1401}\footnote{This is explained in more detail in \cite{1401pap28}.} must be satisfied. This requires that there be only a finite number of such lattices, namely the 34 possibilities found in \cite{1312}. Remarkably, we will see that this physical constraint coincides with a result in the theory of sphere packings. 

\section{Lattice theory}

Lattices arise naturally out of the Cartan classification of Lie algebras. However, by introducing the root system axiomatically, we also saw that lattices can be studied as an independent object in their own right, and we give here an introduction into the basic concepts of lattice theory. 
\subsection{Fundamentals of lattice theory}\hfill

Our first notion is that of a \textit{lattice}:
\begin{definition}[Lattice]
A lattice, $L$, is a \textit{finitely generated free abelian group} with a rank $r\in\mathbb{N}$, that is $L\cong\mathbb{Z}^r$, with a non-degenerate symmetric pairing $(\cdot,\cdot) : L\times L\rightarrow\mathbb{R}$. When expressed as a matrix, the symmetric pairing is termed an \textit{interesection form}.
\end{definition}

This contains the `intuitive definition' as the $\mathbb{Z}^n$ span of $n$ linearly independent basis vectors embedded in a Euclidean space $\mathbb{R}^{n}$: simply map each basis vector $e_i$ to $1\in\mathbb{Z}$, let the group action be addition, and the symmetric pairing the standard Euclidean scalar product. The \textit{ambient space}  bestows a lattice defined in this way with an intersection form given simply by the inner form on that space (for the common case of $R^{n}$ this is just the standard Euclidean scalar product). Of course, it is not necessary to think of the lattice as embedded in any space, but it is often helpful to do so. It can also be convenient to write the basis as a \textit{generating matrix}:
\begin{definition}[Generating matrix]
The generating matrix $M$ of a lattice $L$ is a matrix each of whose rows is a basis vector $e_{i}$, of $L$, and all of the rows of which form a complete basis of $L$:
\begin{equation}
M_{ij}=(e_i)_j
\end{equation}
\end{definition}
We supply definitions of the hypercubic and $ADE$ root lattices in this form, along with a conventional choice of generating matrix in Table \ref{Latticetable}.

\begin{table}
\setlength{\tabcolsep}{0.5em} 
\scalebox{0.6}{
\begin{tabular}{|c|c|c|c|c|}
\hline
Lattice & Ambient space & Condition & Generating matrix & $discr(L)$\\
\hline
\\[-1em]
$\mathbb{Z}^{n}$ & $\mathbb{R}^{n}$ & $x_i\in\mathbb{Z}\text{ }\forall x\in\mathbb{Z}^{n}$ & $\mathds{1}_n$ & $1$\\
\hline
\\[-0.75em]
$A_n, n\geq $ \footnote{Note that the Generating matrix is $n\times n+1$ dimensional} & $\mathbb{R}^{n+1}$ &  $\sum\limits_{i=1}^{n+1}x_i=0\text{ with } x\in\mathbb{Z}^{n+1}$ & $
\begin{psmallmatrix} 
1 & -1 &  0& 0& \dotsm & 0\\
0&1& -1 & 0 & \dotsm & 0\\
\multicolumn{6}{c}{$\vdots$}     \\
0&0& \dotsm & 0&1 & -1  
    \end{psmallmatrix} 
   $ & $n+1$\\[3ex]
\hline
\\[-1em]
$D_n,n\geq 4$\footnote{Note that $D_3=A_3$, as can be seen from the Dynkin diagram. } & $\mathbb{R}^{n}$ & $\sum\limits_{i=1}^{n}x_i=0\mod 2\text{ for } x\in D_n,x\in\mathbb{Z}^{n}$ & $\begin{psm} 
1 & -1 & 0&\dotsm & 0& 0\\
0& 1 & -1 & 0& \dotsm & 0\\
 \multicolumn{6}{c}{$\vdots$}  \\
0 & \dotsm & 0& 1 & -1 & 0\\
0& 0& \dotsm & 0& 1 & -1\\
0&0& \dotsm & 0& 1 & 1
\end{psm}$ & $4$ \\[5ex]
  \hline
\\[-1em]
$E_8$\footnote{The definiton we use is that of the \textit{even coordinate system} of $E8$.} & $\mathbb{R}^8$ & $\sum\limits_{i=1}^{8} x_i = 0\mod 2$, with $\forall x_i\in x(x_i\in\mathbb{Z})$ OR (\textbf{exclusive}) $\forall x_i\in x(x_i\in\mathbb{Z}+\frac{1}{2})$ & $\begin{psm} 
1 & -1 & 0& 0& 0& 0& 0& 0\\
0& 1 & -1 & 0& 0& 0& 0& 0\\
 \multicolumn{8}{c}{$\vdots$}  \\
0& 0& 0& 0& 0& 0& 1 & 1\\
\frac{1}{2} & \frac{1}{2} & \frac{1}{2} & -\frac{1}{2} & -\frac{1}{2} & -\frac{1}{2} & -\frac{1}{2} & \frac{1}{2}
\end{psm}$ & 1\\[4ex]
 \hline
\\[-1em]
$E_7$\footnote{Similarly to $A_n$ this matrix has dimension $7\times 8$} & $\mathbb{R}^8$ & $\{x\in E_8 | (x,v)=0\text{ for a choice of} v\in E_8 \}$ & $\begin{psm} 
1 & -1 & 0& 0& 0& 0& 0& 0\\
0& 1 & -1 & 0& 0& 0& 0& 0\\
 \multicolumn{8}{c}{$\vdots$}  \\
0& 0& 0& 0& 0& 1 & -1& 0\\
\frac{1}{2} & \frac{1}{2} & \frac{1}{2} & -\frac{1}{2} & -\frac{1}{2} & -\frac{1}{2} & -\frac{1}{2} & \frac{1}{2}
\end{psm}$  & 2\\[4ex]
\hline
\\[-1em]
$E_6$ & $\mathbb{R}^8$ & $\{x\in E_8 | (x,v)=0\} \text{ }\forall v\in A_2\text{ for a choice of } A_2\xhookrightarrow{} E_8 \}$ & $\begin{psm} 
1 & -1 & 0& 0& 0& 0& 0& 0\\
0& 1 & -1 & 0&0& 0& 0& 0\\
 \multicolumn{8}{c}{$\vdots$}  \\
0& 0& 0& 0& 1 & -1 & 0& 0\\
\frac{1}{2} & \frac{1}{2} & \frac{1}{2} & -\frac{1}{2} & -\frac{1}{2} & -\frac{1}{2} & -\frac{1}{2} & \frac{1}{2}
\end{psm}$ & 3 \\[4ex]
\hline
\end{tabular}}
\caption{A table showing standard definitions of the hypercubic and even unimodular lattices generated by vectors of norm $2$ (the $ADE$ root lattices). We essentially summarise the definitions given in ch.4 of \cite{CS} and impose the conventions employed in this dissertation.}
\label{Latticetable}
\end{table}
	It is a pleasant and reasurring activity to calculate the Gramian's of the above definitions, and see that they are exactly the Cartan matrices of the corresponding Dynkin Diagrams in Figure \ref{DynkinDiagrams}. The particular choices in Table \ref{Latticetable} have been made for two reasons. First, a theorem due to Witt\cite{Witt}, \cite[Chapter~4]{CS} states that every integral lattice whose generators have norm one or two can be expressed as a direct sum of the lattices in Table \ref{Latticetable}. Second, the gauge groups that we consider in this dissertation occur from the Du-Val singularities which correspond to $ADE$ gauge groups only.\\

\indent To put some colour on the requirment that $L\cong\mathbb{Z}^{r}$, we note that this is equivalent to the lattice being both a finitely generated abelian group and a \textit{torsion free} abelian group.

\begin{definition}[Torsion element]
An element $g$ of a group $G$ with identity element $\mathds{1}$ is said to be a torsion element if there is some $n\in\mathbb{N}$ such that $g^n=\mathds{1}$.
\end{definition}

A group is then said to be \textit{torsion free} if the only torsion element is the identity\footnote{Note that we give the group theoretic definition, which is our concern. An element $m$ of a module $M$ over a ring $R$ is said to be a torsion element if there exists a regular element $r$ of the ring such that $r\circ m=\mathds{1}$}. Our intuitive picture of a lattice once again confirms this requirement: in a Euclidean space, integer combinations of linearly independent vectors can only point away from the additive identity ($\vec{0}$).\\
We will need to distinguish:

\begin{definition}[Integral lattice]
A lattice is said to be integral \textbf{iff} $(x,y)\in\mathbb{Z}\text{ }\forall x,y\in L$.
\end{definition}
\begin{definition}[Even lattice]
A lattice is said to be \textit{even} (\textit{odd})\footnote{An alternative nomenclature sometimes employed is \textit{Type \RNum{2}} and \textit{Type \RNum{1}} respectively.} \textbf{iff} $\text{ }\forall x,y \in L,\text{ }(x,y)\in 2\mathbb{Z}\text{ }(\in 2\mathbb{Z}+1)$.
\end{definition}
\begin{definition}[Unimodular lattice]
A lattice that is both integral and for which $\mathrm{det}(L)=\pm 1$ is said to be \textit{unimodular}.
\end{definition}
A convenient representation of the symmetric pairing of a lattice is it's \textit{Gramian matrix}:
\begin{definition}[Gramian]
Let $L$ be a lattice of rank $r$ with a basis consisting of the vectors $e_i\text{ }; i=1,...,r$. The \textit{Gramian}, $G^{L}$ of $L$ is:
$$G^{L}_{ij}\defeq (e_i,e_j)$$
\end{definition}
Notice that in a Euclidean space $G^{L}=M_{L}^{T}M_{L}$. A useful quantity that will be important to much of what we do is the:
\begin{definition}[Determinant]\label{defdet}
The determinant of a lattice $L$, $\det(L)$,\footnote{Often called the \textit{discriminant}. We prefer this term to avoid confusion with the \textit{discriminant form}.} is the determinant of the intersection form on the lattice. 
\end{definition}
The determinant is linked to the \textit{fundamental region} or \textit{Voronoi cell} of the lattice :
\begin{definition}[Fundamental region]\label{thm: FR}
The \textit{fundamental region} of a lattice $L$ embedded in a space $V$ is the smallest possible subspace $V'\subset V$ in $V$ that tiles the entire space such that every motif contains exactly one lattice point of $L$.
\end{definition}
The connection to the determinant is given in the following statement:
\begin{theorem}
\begin{equation}
|\sqrt{\det(L)}|=(\mathrm{vol}_{FR}(L))
\end{equation}
\end{theorem}
\begin{proof}
Let $\{l_i\}$ be the set of basis vectors of a lattice $L$ of rank $n$ and let $R$ be an orthogonal matrix such that $RR^{T}=\mathds{1}$. Then the Gram-Schmidt procedure guarantees that there will be some $R$, such that $Rl_i=\tilde{l}_i$ leads to an orthogonal basis $\{\tilde{l}_i\}$. In this basis, the volume of the fundamental region is just the volume of the hypercube:
\begin{equation}
\prod\limits_{i=1}^{n}l_i=\det(\tilde{M})=\det(RM)
\end{equation}
From $MM^{T}=G$ have that:
\begin{equation}
\tilde{G}=\tilde{M}^{T}\tilde{M}=M^{T}M\implies\det(\tilde{G})=\det(G)=\big(\det(\tilde{M})\big)^{2}=(\mathrm{vol_{FR}})^{2}
\end{equation}
\end{proof}
$L\cong\mathbb{Z}^r$. The fundamental theorem for finitely generated Abelian groups states that every such group $G$ is isomorphic to a product of cyclic groups of prime power - that is, every finite group admits of a decompositon: 
$$G\cong\mathbb{Z}^n\oplus\bigoplus_{q_i}Z_{q_i}$$
where the $q_i$ are prime and, with $n$, determine $G$. This clearly contains $G\cong\mathbb{Z}^n$ as a limiting case. What then, can we gain from lattice theory that we do not already gain from the theory of finite groups? The answer lies in the extra structure provided by the inner form, and we will be concerned with the inner form on a lattice for much of the rest of this dissertation. We note that it was precisely the Cartan matrix of a root system, which is the inner form on the lattice of simple roots corresponding to a Lie algebra, that characterised the Lie algebras. The importance of the inner form motivates the introduction of a stronger `similarity' requirement than isomorphism:

\begin{definition}[Isometry]
Let $L$ and $L'$ be two lattices.  A map $\phi : L\rightarrow L'$ is said to be an isometry \textbf{iff}:
\begin{enumerate}
	\item  It is an isomorphism $\phi : L\cong L'$ and,
	\item It preserves the inner form on the lattices: $l,m\in L\implies (\phi(l),\phi(m))=(l,m)$
\end{enumerate} 
\end{definition} 
This allows us to introduce the crucial notion of an \textit{embedding}:
\begin{definition}[Embedding]
An embedding of a lattice $M$ into a lattice $L$, denoted $M\xhookrightarrow{\phi}L$, is an injective homomorphism of Abelian groups\footnote{Which is just to say that we preserve the addition of vectors under $\phi : m_1,m_2\in M, \phi(m_1+m_2)=\phi(m_1)+\phi(m_2)$} $\phi : M\rightarrow L$ such that $M$ is isometric to $\phi(M)\subset L$, and the symmetric pairing of $M$ is recovered by restricting the symmetric pairing on $L$ to $\phi(M)\subset L$. We note that two embeddings $M\xhookrightarrow{\phi}L$ and $M'\xhookrightarrow{\phi'}L'$ are said to be \textit{isomorphic} if there exists an isometry $\chi : L\rightarrow L'$ such that $\phi'=\chi\circ\phi : M\rightarrow L'$.
\end{definition}
$\phi(M)$ is then a \textit{sublattice} of $L$. More generally:
\begin{definition}[Sublattice]
A lattice $M$ is said to be a \textit{sublattice} of $L$ with symmetric pairing $(\cdot,\cdot)$ when $M$ is isomorphic to a free abelian subgroup of $L$ and the symmetric pairing is $(\cdot,\cdot)_M : M\times M\rightarrow\mathbb{R} = (\cdot,\cdot)_{\restriction M\times M} $.
\end{definition} 
Embeddings and sublattices are a central concern for us. In particular, we will use the notion of a:
\begin{definition}[Primitive sublattice]
A sublattice $M\subset L$\footnote{We will often use $\subset$ to make explicit the lattice that is embedded into. Often, the most helpful reading of this symbol will be `is a sublattice of'.} is said to be a \textit{primitive sublattice} if the quotient $\frac{L}{M}$ is a torsion free Abelian group - that is, $\mathds{1}_{\frac{L}{M}}$ is the only torsion element. Naturally, a lattice embedding is said to be \textit{primitive} if $\phi(M)$ is a primitive sublattice of $L$.
\end{definition}
For the purposes of computation, a more `practical' definition is useful. Recall that $\frac{L}{M}$ introduces an equivalence relation $l_m$ which identifies all the vectors in:  $\{l+m | m\in M\}$. It is sometimes said that the quotient is `$L\mod M$'. Thus, an element $l\in L,\not\in M$ but such that there is some integer $k$ for which $kL\in M$, is a torsion element or order $k$ since we must identify these two in the quotient. An equivalent way of stating this is to say that there is some non-integer rational $q$ for which $qm\in L$. Thus we have the following useful equivalence:
\begin{equation}\label{eq: primitive_def}
(M\otimes\mathbb{R})\cap L = M\iff \mathrm{Tor}(\frac{L}{M})=\mathds{1}
\end{equation}
An alternative terminology sometimes used is that of a \textit{section}. A $k$-dimensional section $L_{k}$ of a lattice $L$ is defined as: $L_k\defeq\mathbb{R}^{k}\cap L$. Since any lattice of dimension $k$ must have lattice basis vectors which over $\mathbb{R}$ span a k-dimensional Euclidean space, this definition is equivalent to \ref{eq: primitive_def}.
Perhaps the simplest example of a non-primitive embedding is shown in Figure \ref{torsionpicture}. A more interesting example is Figure \ref{BCC_torsion}. 
\begin{figure}
\includegraphics[ scale=0.3]{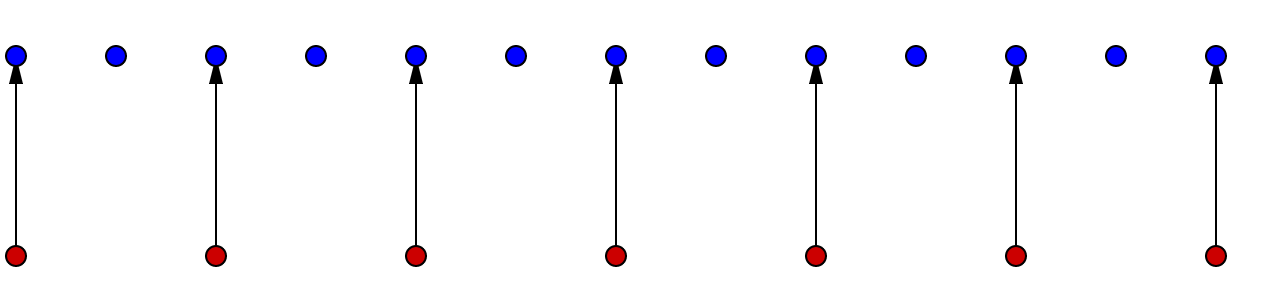}
\caption{A simple example of torsion, in this case the torsion group is $\mathbb{Z}_2$. The lattice consisting of the red circles, $L_1$, is embedded into the lattice consisting of the blue circles, $L_2$. The embedding is non-primitive, since every second lattice vector, $l_2$ say, is missed out and twice such a lattice vector results in an element of the embedded lattice. Hence, in the quotient $\frac{L_2}{L_1}$, $l_{2}^{2}=\mathds{1}$.}
\label{torsionpicture}
\end{figure}
\begin{figure}
\includegraphics[scale=0.7]{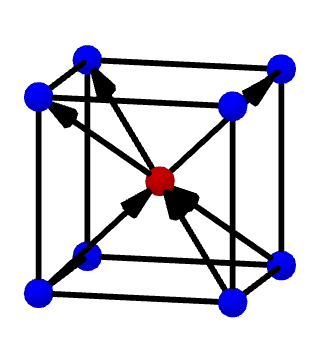}
\caption{Here we show torsion occuring from an embedding of the simple cubic lattice (blue) into the BCC lattice (blue and red). The vectors are possible torsion vectors, and the torsion group is once again $\mathbb{Z}_2$.The central point is a \textit{deep hole} - a point furthest away from any lattice point, of the simple cubic lattice. The $E_8$ lattice can be constructed analogously, by adding in such a central point to the $D_8$ lattice, whereupon we say that $D_8$ is \textit{self-glued} to form $E_8$.  A more involved construction allows one to recover all 23 Niemeier lattices from the deep holes of the Leech \cite[Chapter~23]{CS}.}
\label{BCC_torsion}
\end{figure}

Another useful construction is the \textit{orthogonal complement}, $M^{\perp}$, of a lattice $M$ embedded into another lattice $L$, defined as:
\begin{definition}[Orthogonal complement]
$$M^{\perp}\defeq \{l\in L | (l,m)=0\text{ }\forall m\in M\}$$
\end{definition}
\begin{remark}\label{rem:orthprimitive}
We note that the orthogonal complement of a sublattice $M$ of $L$ is \textit{always} primitive. A proof is had by contradiction: Assume that $M^{\perp}$ is not primitive. Then there is some rational $q$ such that for $m^{\perp}\in M^{\perp}$, $qm^{\perp}\in L\not\in M^{\perp}$. But by linearity $(m^{\perp},m)=0\implies(qm^{\perp},m)=0)$, so the previous two conditions cannot be simulatenously satisfied.
\end{remark}
For convenience, we also define something close to the opposite of a sub-lattice:
\begin{definition}[Overlattice]
A lattice $L$ is said to be an \textit{overlattice} of a lattice $M$ if $M$ is a sublattice of L with the index $[L:M]$ being finite. 
\end{definition}
\subsection{Discriminant Group and form}\hfill

 The canonical concern of classical lattice theory is to solve the following three problems:
\begin{enumerate}
	\item The classification of lattices.
	\item The existence of an embedding of one lattice into another.
	\item The uniqueness of a lattice embedding.
\end{enumerate}

These problems have essentially been solved for the case of even and odd \textit{unimodular} lattices \cite{Nikulin}. The \textit{discriminat group} and the \textit{discriminant form} are essential tools for strengthening these results to non-unimodular lattices. Since we will need the more general case to develop the Kneser-Nishiyama method, we introduce them now. 
	First, we require the definition of the \textit{dual lattice}:
\begin{definition}[Dual lattice]
A lattice $L^{*}$ is said to be \textit{dual} to a non-degenerate integral lattice $L$ of rank $r$ whose basis vectors span the Euclidean space $R^{r}$ \textbf{iff}
$$L^{*}=\{x\in\mathbb{R}^{r} | (x,y)\in\mathbb{Z}\text{ }\forall y\in L \}$$
A lattice is said to be \textit{self-dual} \textbf{iff} $L=L^{*}$.
\end{definition}
If $\{e_{1},\cdots,e_{n}\}$ is a basis for a lattice $L$, then a dual basis $\{e_{1}^{*},\cdots,e_{n}^{*}\}$ can be taken so as to satisfy $e_{i}^{*}\cdot e_{j}^{*}=\delta_{ij}$ since then $l\cdot l^{*}\in\mathbb{Z}$ for any $l\in L, l^{*}\in L^{*}$. Alternately put:
\begin{equation}
G^{*}G=\mathds{1}\iff (G^{*})^{-1}=G
\end{equation}
which in turn yields:
\begin{theorem}\label{thm: invdets}
\begin{equation}
\det(L)=\frac{1}{\det(L^{*})}\iff\mathrm{vol_{FR}}(L)=\frac{1}{\mathrm{vol_{FR}}(L^{*})}
\end{equation}
\end{theorem}
We further have the result that \cite[Lemma~1.1]{Nishiyama}:
\begin{theorem}\label{thm: torsiondet}
For a lattice given as $L=M|K$, we have that $\det(L)=\frac{\det(M)}{|K|^2}$
\end{theorem}
We note that the number of cosets of $M$ in $L$ is exactly the order of the group $K$. In particular, if $L$ is the dual lattice :
\begin{corollary}\label{orderdet}
\begin{equation}
[L^{*}:L]=|\det(L)|
\end{equation}
\end{corollary}
\begin{proof}
We have that $\det(L^{*})=\frac{det(L)}{([L^{*}:L])^{2}}$ from Thoerem \ref{thm: torsiondet}, and combining this with Theorem \ref{thm: invdets} shows that $(\det(L))^{2}=([L^{*}:L])^2\implies\det(L)=[L^{*}:L]$
\end{proof}
We will later interpret $[L^{*}:L]$ as the order of the \textit{discriminant group}.\footnote{Note that since a lattice is an \textit{abelian} group the subgroups are normal, and so the index $[M:N]=|\frac{M}{N}|$.}
To develop an intuition for the dual lattice\footnote{An informative description of the dual lattice is given in \cite[Lecture~8]{Regevlattices}.}, consider the following argument. Assume that $L\neq L^{*}$. For example, suppose that $L$ is simply $\mathrm{span}_{\mathbb{Z}}(e_1+e_2)$ in $\mathbb{R}^2$. Such a lattice is clearly not equal to it's dual, which can be taken as $\mathrm{span}_{\mathbb{Z}}(e_1,e_2)$. This last lattice is `finer' than $L$ - it misses less points in the underlying space. Thus, the condition $L=L^{*}$ can be understood as a maximality condition. The self dual lattices are as `fine' as possible. \\
\indent What does it mean for a lattice to be fine? Recall that the fundamental region of a lattice is the smallest possible subvolume of the lattice which can tile the entire lattice with exactly one lattice point in each tile. We would expect therefore that if an integral lattice is as fine as possible, if it is self-dual that is, then it's fundamental region should be minimal. Since for an integral lattice the Gramian consists of integers, we must have that $| \mathrm{det}(L) | = | \mathrm{det}(G) |\geq1$. 
\begin{theorem}\label{thm: uni-s.d}
For an integral lattice $L$:
\begin{equation}
L=L^{*}\iff |\det(L)|=1
\end{equation}
\end{theorem}
\begin{proof}
Observe that $L\subset L^{*}$ always. This entails that $\det(L)\leq\det(L^{*})$, since $L^{*}$ must contain all points in $L$ and may contain more, allowing a smaller fundamental region. Moreover, if $L\neq L^{*}$, then the fundamental region of $L^{*}$ must be smaller than that of $L$. This is because the latter will tile the entire ambient space common to both with exactly one lattice point of $L$ in each motif, and by assumption $L^{*}$ has at least one more lattice point, which entails that there is some motif of the fundamental region containing more than one lattice point of $L^{*}$. So we have that: $\det(L)=\det(L^{*})\iff L=L^{*}$. From Theorem \ref{thm: invdets}, we have that $L=L^{*}\iff |\det(L)|^{2}=1\iff |\det(L)|=1$.




\end{proof}
The dual lattice allows us to give the definition of the \textit{discriminant group}:
\begin{definition}[Discriminant group]
For an even lattice $L$, the \textit{discriminant group} is the quotient:
$$G_L\defeq\frac{L^{*}}{L}$$
\end{definition}
A very simple calculation of the discriminant group is afforded by $A_1$. $A_1$ can be taken as $\mathrm{span}_{\mathbb{Z}}\{(1,-1)\}$. The dual is then just $\frac{1}{2}(1,-1)$. Notice that if $\{\vec{b_1}\hdots,\vec{b_r}\}$ is basis for $L$,  $L^{*}$ must be a subset of the Euclidean space $\mathrm{span}_{\mathbb{R}}\{\vec{b_1}\hdots,\vec{b_r}\}$, so that we cannot, for example, have $(0,1)\in A_{1}^{*}$.
\indent A more involved computation can be carried out for $E_7$. For convenience we work in $R^{8}$ so that a generating matrix for $E_7$ is that in Table \ref{Latticetable}. We now need to find a set of vectors that give integer inner form with all of the basis vectors, lie in the subspace of $R^{8}$ spanned by the basis vectors of $E_7$, and are not in $E_7$. There is only one such vector and it is: $v=\begin{psm}(\frac{1}{4})^6,((-\frac{3}{4})^{2})\end{psm}$. To see that only one (non-unique) vector may be chosen, observe that $\det(E_7)=2$, and use the fact that the order of the quotient group $\frac{L^{*}}{L}$ is the lattice determinant (Corollary \ref{orderdet}).  Since, the quotient group $G_{E_7}=\frac{E_{7}^{*}}{E_7}$ always contains the identity element, there is only one other element. Hence, we have that $(E_7)^{*}$ is generated by the overcomplete\footnote{Over $\mathbb{R}^n$, but not over $\mathbb{Z}^n$.} basis given by $M_{E_7}+\mathrm{span}_{\mathbb{Z}}\{v\}$.
\indent The discriminant form is then a quadratic form on the discriminant group:
\begin{definition}[Discriminant form]
For a lattice $L$ with lattice points $l\in L$ and discriminant group $G_L$, the discriminant form, $q_L$, is defined as:
\begin{equation}
\begin{aligned}
& q_L : G_L\rightarrow\frac{\mathbb{Q}}{2\mathbb{Z}}\\
& q(x+l)\equiv q(x)\mod 2\mathbb{Z}
\end{aligned}
\end{equation}
which maps:
\begin{equation}
l\rightarrow l^2
\end{equation}
\end{definition}
The discriminant form acts as a natural `inner form' on the discriminant group induced by the inner form of $L^{*}$. To see this, consider some $l^{*}\in L^{*}$ and some $l\in L$. This $l^{*}$ is a  \textit{representative} of the equivalence class $[l^{*}]=  \{l^{*} + l | l\in L\}$ in the discriminant group.  $L^{*}$ contains an inner form, so we may take $q=(l^{*},l^{*})\in\mathbb{Q}$. We must have:
\begin{equation}
q(l^{*})\equiv q'(l^{*}+l)=(l^{*},l^{*})+2(l^{*},l)+(l,l))=q+Z,\text{ with }Z\in 2\mathbb{Z}
\end{equation}
where the last equality follows from the definition of the dual lattice. Thus, introducing the mod $2\mathbb{Z}$ equivalence as required by our definition, we see that we can calculate the discriminant form for a lattice $L$ if we know $L^{*}$. We can use the discriminant form to introduce a bilinear form between two elements of $G_L$ as:
\begin{definition}[Bilinear form]
The \textbf{bilinear form} on the discriminant group is a map $b(x,y): G_L\times G_L\rightarrow \frac{\mathbb{Q}}{\mathbb{Z}}$ defined by:
\begin{equation}
b(x,y)=\frac{1}{2}(q(x+y,x+y)-q(x,x)-q(y,y))\mod\mathbb{Z}
\end{equation}
\end{definition}
Where the definition is mod $\mathbb{Z}$ for the same reason as the discriminant form.
\subsection{Classification theorems for non-unimodular lattices}\hfill

The discriminant group and the discriminant form are at the heart of Nikulin's seminal work\cite{Nikulin}\footnote{Although for the readers convenience we mostly follow the notation of \cite{1312}}. In this section we will give the four results from that paper which are essential to the Kneser-Nishiyama method. 


\begin{proposition} (\cite[Proposition~1.4.1]{Nikulin})]\label{Nik1}
Let $M$ be an even lattice, and let $H$ be a subgroup of $G_M$. An even sublattice, $(M;H)$ of $M^{*}$ which itself has $M$ as a sublattice:
\begin{equation}
(M;H)\defeq \{m^{*}\in M^{*} | (m^{*}\text{ mod }M)\in H\}
\end{equation}
can be constructed from $H$ \textbf{iff} there is a $H$ such that $q_{M\restriction H}=0$. 
\end{proposition}
We will mostly use the ``$\impliedby$'' direction of this theorem to introduce just such an even lattice $(M;H)$. Notice that if we let $m_{1}^{*}\in M$, then $m_{1}^{*}\text{ mod }M=\mathds{1}\in H$ always, so we have $M\subseteq (M;H)$ as required. 


\begin{proposition} \cite[Proposition~1.6.1]{Nikulin}\label{Nik2}
Let $M$ and $N$ be even lattices. There exists an isomorphism $\gamma : G_M\cong G_N$ with $q_M=-q_N\circ\gamma$ \textbf{iff} $M$ is a primitive sublattice of an even unimodular lattice $L$ with $N$ given by $N=M^{\perp}\subset L$. 
\end{proposition}
\begin{observation}
If $M$ and $N$ have discriminant groups related by the isomorphism $\gamma$, then:
$$H_{M\times N}\defeq\{(m',\gamma(m'))\in G_M\times G_N | \forall m'\in G_M\}$$
is an isotropic subgroup of $G_M\times G_N$, which appears in the construction of the even unimodular lattice  $L=(M\oplus N);H_{M\times N}$.
\end{observation}
\begin{proof}
That $H_{M\times N}$ is isotropic follows from the fact that for $h\in H_{M\times N}$, $q(h)=q((m,n))=q(m)+q(\gamma(m))=q(m)-q(m)=0$. From Theorem \ref{thm: torsiondet}, $\det(L)=\frac{\det(M)\det(N)}{|H_{M\times N}|^{2}}$. But, $G_M\cong G_N\implies |H_{M\times N}|=|G_{M}|=|G_{N}|$. From Corollary \ref{orderdet} $\det(M)=|G_M|$, for general $G_M$. Hence, $\det(L)=1$. 
\end{proof}
Thus we can use Propositions \ref{Nik1} and \ref{Nik2} in conjunction. The latter states necessary and sufficient conditions for the existence of an isotropic subgroup $H_{M\times N}$ of two even lattices $M$ and $N$, and the former uses this to introduces an even unimodular lattice $L=(M\oplus N);H_{M\times N}$.\\

	Before stating the final two theorems, we require the definition of the \textit{signature} of a lattice:
\begin{definition}[Signature of a lattice]
Let $L$ be a non-degenerate lattice with symmetric pairing $(\cdot,\cdot)$. Then, if $r_+$ and $r_-$ are, respectively, the number of positive and negative eigenvalues, then the \textit{signature} of $L$ is the pair $(r_+,r_-)$. A common notation is to write a lattice $L$ of signature $(r_+,r_-)$ as $L^{(r_+,r_-)}$.
\end{definition}

Nikulin also showed that (\cite[Theorem 1.12.4]{Nikulin}):
\begin{theorem}\label{Nik3}
A sufficient condition for the existence of a primitive embedding of an even lattice M of signature $(m_+,m_-)$ into \textit{some} even unimodular lattice $L$ of signature $(l_+,l_-)$, $M\xhookrightarrow{\phi}L$ is: $\rank(M)\leq\frac{\rank(L)}{2}$ and $l_+=l_- \mod 8$ 
\end{theorem}
This theorem establishes only that there is a primitive embedding into \textit{some} even unimodular lattice. We will be considering this for the rank 24 even unimodular lattices. There are 24 of these, and an embedding is only certain to exist into one of them. This theorem can be strenghtened to show (\cite[Theorem 1.14.4]{Nikulin}):
\begin{theorem} \label{Nik4}
A sufficient condition for the existence \textit{and} uniqueness of a primitive embedding of an even integral lattice $M$ into an even unimodular lattice $L$ is that: $l_{+}>M_{+}$, $l_{-}>m_{-}$ and that the least possible number of generators of the discriminant group, $N(G_M)\leq rank(L)-rank(M)-2$.
\end{theorem}
\section{The Kneser-Nishiyama method}
We saw in Section \ref{Ftheory} that the non-abelian low energy energy effective physics of an interesting and tractable class of F-theory models is contained within the frame lattice $W$. The Kneser-Nishiyama method gives a way to find this lattice, and in this section we give an outline of it's implementation.\footnote{For the convenience of the reader who wishes to find more details in the references, the discussion here will be in negative signature, since this is standard in those articles. In this dissertation, we \textbf{do not} subsequently work in negative signature.}
\subsection{The direct approach}\hfill

 In the $K3$ surfaces which we study, the surface comes with a complex structure, and in particular can be used to define a signature $(1,\rho_X-1)$ lattice called the \textit{Neron-Severi lattice}.
The Neron-Severi lattice admits of an orthogonal decomposition:
\begin{equation}\label{eq:S,U,W}
S_X=U\oplus W
\end{equation}
Under this `direct' approach, to find the physics contained in $W$, one would find the Neron-Severi lattice defined by the $K3$ surface $X$ by calculating it's orthogonal decomposition into a hyperbolic plane lattice and the frame lattice. This approach, however, is technically difficult because we are not working in a space of definite signature. Since in thisdissertation we consider the \textit{extremal} $K3$ \textit{surfaces} - that is, those with $\rho_X=20$, and $W$ is 18 dimensional with signature $(0,18)$. 
\subsection{The indirect approach}\hfill

\label{Indirectapproach}
We want to work in definite signature. A way to do this is found in \cite{Nishiyama} and we give a quick review here. \\
\indent Recalling that $L^{(3,19)}$ is isometric to $U^{\oplus3}\oplus (-E_8)^{\oplus 2}$ we have:
\begin{definition}[Transcendental lattice]
The transcendental lattice for a given $K3$ surface $X$ is defined as:
$T_X=\{l\in L^{(3,19)} | (l,s)=0\text{ }\forall s\in S_X\}$
\end{definition}

$T_X$ has signature $(2,20-\rho_X)$, and so the lattice obtained by taking the negative of all of its elements, $-T$ has signature $(20-\rho_X,2)$. By application of Proposition \ref{Nik2} to this definition we see that $G_{T_X}\cong S_X$.\\
\indent Consider the lattice $E_8\oplus U^{\oplus 20-\rho_X}$. This has signature $(20-\rho,28-\rho_X)$ and in our case where $\rho_X=20$, this lattice is just the negative definite $E_8$. We define the lattice $T_0$ as:
\begin{definition}[$T_0$]
\begin{equation}
-T_0\defeq (-T^{\perp}\subset -E_8\oplus U^{\oplus 20-\rho_X})
\end{equation}
for a primitive embedding $T\xhookrightarrow{}-E_8\oplus U^{\oplus 20-\rho_X}$.
\end{definition}
We notice that: $\frac{1}{2}\rank(\mathrm{overlattice})=\frac{8+2(20-\rho_X)}{2}=4+(20-\rho_X)\geq (2+20-\rho_X)$; $20-\rho_X=28-\rho_X\mod 8$; $20-\rho_X\leq20-\rho_X$ and $2\leq 28-\rho_X$. Thus, by Theorem \ref{Nik3}, such an embedding always exists and can be used to form a signature $(0,26-\rho_X)$, for us $(0,6)$, lattice. \\
\indent The next stage is then to primitively embed $-T_0$ into a negative definite even unimodular lattice of rank 24. There are 24 such lattice, these being the 23 \textit{Niemeier lattices}, and the Leech lattice\footnote{It is conventional to not include the Leech lattice under the name ``Niemeier lattice".}. Nishiyama then showed that the frame lattice $W$ is the orthogonal complement of $T_0$ in one of the even unimodular rank 24 lattices $\mathcal{N}$:
\begin{theorem}
\begin{equation}
W=(T_{0}^{\perp}\subset\mathcal{N})
\end{equation}
\end{theorem}
Specialising to the case $\rho_X=20$ we again have: $6\leq 12$, $24=0\mod 8$, $0\leq 0$ and $6\leq 24$ so such an embedding always exists into \textit{some} Niemeier lattice although not necessarily into every Niemeier lattice. \\
\indent In general, different embeddings of $T$ will lead to different $T_0$. For a given choice of primitive embeddings, however, the list of $W$ lattices produced will be the same \cite[pg.19]{1312}. Since we only wish to find the $W$ lattices, we may choose any set of primitive embeddings $T\xhookrightarrow{}E_8$.
\subsection{Extracting Physics}\hfill

We cannot guarantee that such an embedding is unique, so the picture that results is the following: for every primitively embedding of $T_0\xhookrightarrow{}\mathcal{N}_i$ we generate some $W_{\mu}^{T_0}$ for \textit{all possible} primitive embeddings $T_0\xhookrightarrow{}\mathcal{N}_i$ into \textit{all possible} $\mathcal{N}_i$. Just one $T_0$ will generate a set of frame lattices $W^{T_0}=\{W_{\mu}\}$, and the set of all possible $W$'s is the union of all such sets $W^{T_0}$. The classification of elliptic fibrations on our choice of  $CY_4=K3\times K3$, and hence the possible non-abelian massless scalar fields of these F-theory models,  is given in all of these frame lattices $W$.
\section{Computing $T_0$}
The first task, then, is to compute the  $\text{ }T_0\defeq T^{\perp}\subset E_8$ lattices for the list of transcendental lattices, $T$, in \cite{1312}. We adopt the notation that $T_{[abc]}$ corresponds to a lattice with Gramian:
$$T_{[abc]}\estimates
\begin{pmatrix}
2a & b \\
b & 2c
\end{pmatrix} 
$$
The procedure for finding $T_0$ is as follows:
\begin{enumerate}
	\item Find a \textit{primitive} embedding of $T_{[abc]}$ into $E_8$.
	\item Find a $\mathbb{Z}$-submodule orthogonal to $T_{[abc]}$ in $E_8$.
	\item Check that this $\mathbb{Z}$-submodule is torsion free in $E_8$.  
\end{enumerate}
The diagonal entries of the root lattice correspond to the norms of the generating vectors, and the off-diagonal elements are just the result of their dot product. This imposes strong constraints on the possible choice of a primitive sublattice since the embedding, $\phi$, must be an isometry. There are, for example, only $2160$ norm $4$ roots in $E_8$, and we would have to choose twice from this set to find a primitive embedding of $T_{[111]}$.\\
\indent Once a primitive embedding has been chosen, our procedure is to look for six linearly independent vectors, denoted $\tilde{e}_i$, that span the orthogonal complement of the primitively embedded lattice. We adopt the notation that $L|K$\footnote{In \cite{1312} the notation $L;K$ is used - we do not follow this usage to avoid confusion with the notation of \ref{Nik1} where the quotient is with the dual lattice of $M$.} denotes that, in the chosen basis, the lattice is given by $L$, with quotient $\frac{T_0}{L}=K$.\footnote{This does \textbf{not} mean that $T_0$ has torsion - it is just an artefact of the basis we have chosen. This notation is more convenient for our purpose essentially due to Witt's theorem, since our final use for these lattices will be to embed them into the Niemeier lattices, all but one of which are direct sums of the A-D-E root lattices.}. Since the transcendental lattice $T_{[abc]}$  is primtively embedded in $E_8$, the ``$\impliedby$" direction of Proposition \ref{Nik2} entails that $T_{[abc]}$ and $T_0$ have isomorphic discrimininant groups. The theorem (Proposition 1.2 of \cite{LandCWE}):
\begin{theorem}\label{rem:dets}
If $L$ is a unimodular lattice, with $M$ and $N$ primitively embedded mutually othogonal sublattices in $L$, $M\oplus N\subseteq L$, then $\det(M)=\det(N)$
\end{theorem}
affords us a check on our result when finished. To do this check we will often employ Theorem \ref{thm: torsiondet}.
We now give two examples of this procedure - the first is a simpler case, and the second slightly more involved.
\subsubsection{$T_{[111]}$}
Our first example is the lattice $T_{[111]}\estimates\begin{psm}2 & 1\\1& 2 \end{psm}$. The first step requires us to find a primitive embedding: $T_{[111]}\xhookrightarrow{} E_8$. Practically, this works out to choosing two norm two vector with the appropriate inner form (here 1) that are primitive in $E_8$. The following remark will be helpful in this:
\begin{remark}\label{thm:1/2_1}
Any sublattice $T_p\subset E_8$ spanned by two \textit{individually} primitive vectors $\phi(e_1)=\tilde{e}_1$ and $\phi(e_2)=\tilde{e}_2$ such that $\tilde{e}_1\in \{(x_1,\hdots, x_8)| x_i\in\mathbb{Z}\}$,  $\tilde{e}_2\in \{(x_1,\hdots, x_8)| x_i\in\mathbb{Z}+\frac{1}{2}\}$ with some component $k$ for which $(\tilde{e_1})_k=0$ and $(\tilde{e_2})_k=\frac{1}{2}$ is a \textit{primitive sublattice} of $E_8$.
\end{remark}
\begin{proof}
Assume for contradiction that $T_p$ is not a primitive sublattice, and let $z_i\in\mathbb{Z}$. Then, from equation \ref{eq: primitive_def} there is some $q\in\mathbb{Q}$ and some $t=z_1\tilde{e}_1+z_2\tilde{e}_2\in T_p$ such that $qt\not\in T_p$ but $qt\in E_8$. Then, $(qt)_k=q\frac{z_2}{2}$. From the definition of $E_8$, $qz_2\overset{!}{\in}\mathbb{Z}\implies qz_2\tilde{e}_2\in E_8$ always. But then the lattice can only fail to be primitive if $qz_1\tilde{e}_1\in E_8,\not\in T_p$ - which would contradict the primitivity of $\tilde{e}_1$. Hence, there is no such $q$ and $T_p$ is a primitive sublattice of $E8$.
\end{proof}
This remark allows us to reduce the problem of finding a primitive \textit{sublattice} to the simpler one of finding two individually primitive \textit{vectors}, so long as we choose from the two different set of vectors allowed in $E_8$. \\
\indent The norm-2 vectors of $E_8$ are it's $240$ roots. They are the permutations of the vectors satisfying the conditions in Table \ref{Latticetable}\footnote{We abuse notation slightly to allow arbitrary sign in each $\frac{1}{2}$ term in the second vector.}:
\vspace{-2mm}
\begin{equation}
(\pm1,\pm 1,0,0,0,0,0,0),\text{  },
\bigg(\pm\frac{1}{2},\pm\frac{1}{2},\pm\frac{1}{2},\pm\frac{1}{2},\pm\frac{1}{2},\pm\frac{1}{2},\pm\frac{1}{2},\pm\frac{1}{2}\bigg)
\end{equation}

There are many choices of primitive embedding, each leading to a different $T_0$, but from the discussion at the end of subsection \ref{Indirectapproach} \textit{any} embedding is sufficient for us. Let us choose:
\vspace{-2mm}
\begin{equation}
e'_1=(1,1,0,0,0,0,0,0),\text{  }e'_2=\bigg(\frac{1}{2},\frac{1}{2},\frac{1}{2},\frac{1}{2},\frac{1}{2},\frac{1}{2},\frac{1}{2},\frac{1}{2}\bigg)
\end{equation}

Remark \ref{thm:1/2_1} guarantees that this embedding is primitive. \\
\indent To aid us in finding the orthogonal complement, we make the observation that the contrapositive of Theorem \ref{rem:dets} states that if the determinants of a sublattice of $M\subset E_8$ and $T_{[111]}$ do not agree, then they cannot be orthogonal complements. This will be a helpful guide in choosing an ansatz for the orthogonal complement.\\
\indent There are six zeros in $e'_1$, $\det(T_{[111]}=3)$ and $\det(E_6)=3$. The first six zeros allow us to find an $A_5$ sublattice as follows:\footnote{We adopt the standard notation that $a^{n}$ inside a vector denotes $n$ copies of the entry $a$: $\begin{psm} 1,&2^{3},&4\end{psm}=\begin{psm}1,&2,&2,&2,&4\end{psm}$}
\begin{equation}
(0^2,1,-1,\text{}0^4),\text{ }(0^3,1,-1,0^4),\text{ }\text{ }(0^4,1,-1,0^2),\text{ }(0^5,1,-1,0),\text{ }(0^6,1,1)
\end{equation}
Seeing that $\det(E_6)=\det(T_{[111]})$, and noting that adding any vector consisting of $1/2$'s will keep the lattice primitive (again from Remark \ref{thm:1/2_1}), we are motivated to add the final ansatz vector $\tilde{e}_6=\begin{psm}-\frac{1}{2},&\frac{1}{2},&-\frac{1}{2}&-\frac{1}{2}&-\frac{1}{2},&\frac{1}{2},&\frac{1}{2},&\frac{1}{2}\end{psm}$. To check our result, we compute the Gramian of this lattice:
\begin{equation}
G^{T_0}=\begin{psm}
2&-1&0&0&0&0\\
-1&2&-1&0&0&0\\
0&-1&2&-1&0&-1\\
0&0&-1&2&-1&0\\
0&0&0&-1&2&0\\
0&0&-1&0&0&2
\end{psm}
\implies \det(T_0)=3
\end{equation}
as required.\\
\indent Of course, this result simply follows from the definition of $E_6$ as the orthogonal complement of $A_2\cong T_{[111]}$ in $E_8$. The above calculation was simply to show how to employ the procedure for finding $T_0$ and the steps involved. So armed, we can move on to extending the known results in Table 4.2 of \cite{1312} with the next example of $T_{[555]}$.
\subsubsection{$T_{[555]}$}
We have:
\begin{equation}
T_{[555]}\estimates\begin{pmatrix}
10&5\\
5&10
\end{pmatrix}\implies\det(T_{[555]})=75
\end{equation}
Carrying out Step 1 as before, we list the norm-10 vectors in $E_8$. They are the permutations of:
\begin{equation}
(\pm 3,\pm 1, 0^6),\text{ }(\pm 2,\pm 2,\pm 1,\pm 1,0^4),\text{ }(\pm2, (\pm 1)^6,0),\text{ } \bigg(\pm\frac{5}{2},\pm\frac{3}{2},(\pm\frac{1}{2})^{6}\bigg),\text{ }\bigg((\pm\frac{3}{2})^{4},(\frac{1}{2})^{4})\bigg)
\end{equation}
satisfying the condition to be in $E_8$ given in Table \ref{Latticetable}. There are many choices that will lead to a primitive embedding but we see from Table \ref{Latticetable} that $\det(A_4)=5$ which is a prime factor of $\det(T_{[555]})=75$. This prompts us to look for a primitive embedding that will have $A_4$ as a sublattice of the orthogonal complement. Practically, this means we look for vectors with five of the components equal.  Remark \ref{thm:1/2_1} entitles us to choose:
\begin{equation}
e'_1=(1, 3,0^{6}),\text{ } e'_2=\bigg(\frac{1}{2},\frac{3}{2},\frac{5}{2},(-\frac{1}{2})^5\bigg)
\end{equation}
The $A_4$ lattice is then found as:
\begin{equation}
\tilde{e}_1=(0^3,1,-1,0^3),\text{  }\tilde{e}_2=(0^4,1,-1,0^2),\text{  }\tilde{e}_3=(0^5,1,-1,0),\text{  }\tilde{e}_4=(0^6,1,-1)
\end{equation}
The only other root lattice with determinant equal to a prime factor of $75$ with the appropriate rank is $A_2$, but this would take up the last two vectors and $5\times 3 =15\neq 75$. So we look for two more vectors orthogonal\footnote{Again, we emphasise that this is just a helpful ansatz, it is not a requirement.} to all of the others, such that their norm contains a prime factor of $75$. At this point the code given in Figure \ref{orthcode} can be useful. The fact that the first two components of $e'_1$ and $e'_2$ are the same up to scaling allows us to find a norm-10 vector as: 
\begin{equation}
\tilde{e}_5=(3,-1,0^6) 
\end{equation}
We do not yet have a factor of $3$, and the lowest norm vector with this factor has norm 6. There is just such a vector orthogonal to all of the ones we have chosen so far, namely:
\begin{equation}
\tilde{e}_6=(0^2,1^6)
\end{equation}
Such a vector does not contain all of the vectors in $T_0$ - we miss every second vector when taking the sum $\tilde{e}_5+\tilde{e}_6=(3,-1,1^6)$ since $\frac{1}{2}(\tilde{e}_5+\tilde{e}_6)=(\frac{3}{2},-\frac{1}{2},(\frac{1}{2})^6)\in E_8$. Thus, we also need to add in these elements which corresponds to a group $\mathbb{Z}_2$ for the cosets of $\mathrm{span}_{\mathbb{Z}}\{\tilde{e}_1,\hdots,\tilde{e}_6\}$ in $T_0$. Thus, the final $T_0$ is given as:
\begin{equation}
T_{0}=A_4\oplus(10)\oplus(6)|\mathbb{Z}_2
\end{equation}
Using Theorem \ref{thm: torsiondet} we check that $\det(T_0)=75$, as required. 

\subsubsection{Using SAGE}
The advantage of doing these computations by hand is that we can choose a basis for $T_0$ that makes the sub-root lattices in it explicit. 23 of the Niemeier lattices $\mathcal{N}$ are expressible as a direct sum of root lattices, and since our final task is to carry out $T_0\xhookrightarrow{}\mathcal{N}$, it is advantageous to express $T_0$ in this form. \\
\indent It is also possible, however, to calculate these embeddings using computer software. This has been done by the author in the SAGE program, the code for which is given in the Supplementary Information as Figure \ref{fig:SAGECODE}. Using both of these methods, we can compute all of the unsolved cases in Table 2 of \cite{1312}. The lattices calculated in that paper and this one by hand are shown in Table \ref{tab:T0}. Those calculated in SAGE are given in \ref{tab:T0SAGE}. The calculation shows that every $T_0$ contains a root. \\
\indent This is significant because our final task is to primitively embed the $T_0$ lattices into the Niemeier lattices, and determine whether the orthogonal complemenent contains a root. There is one Niemeier lattice with no roots, the \textit{Leech lattice}. Tables \ref{tab:T0} and \ref{tab:T0SAGE} tell us that there is no $T_0$ embeddable into the Leech lattice since every $T_0$ has a root.\footnote{For this choice of primitive embedding. There may be other choices of primitive embeddings of a $T$ that contains a $T_0$ embeddable into the Leech lattice.} This gives us a practical benefit, namely: that we need not consider embeddings into the Leech lattice and so only work with the 23 Niemeier lattices for this choice of $T_0$. It also removes a simple way to find a $W$ with no roots - since if it where possible to embed a $T_0$ lattice into the Leech lattice, the orthogonal complement could not, a fortriori, contain any roots.

\begin{table}[h]
\begin{center}
\scalebox{1.0}{
\setlength{\tabcolsep}{0.5em} 
{\renewcommand{\arraystretch}{1.5}
\begin{tabular}{|c|c|c|c|}
\hline
$[a\text{ }b\text{ }c]$ & $T_0$ & Contains root? & $\det(T_{[abc]})=\det(T_0)$\\
\hline
${[101]}$ & $D_6$ & Yes & 4 \\
\hline
${[111]}$ & $E_6$ & Yes & 3 \\
\hline
${[201]}$ & $D_5\oplus A_1$ & Yes & 8\\
\hline
${[202]}$ & $A_3\oplus A_3$ & Yes & 16 \\
\hline
${[222]}$ & $A_5\oplus(10)|\mathbb{Z}_2$ & Yes & 12 \\
\hline
${[301]}$ & $A_5\oplus A_1$ & Yes & 12 \\
\hline
${[311]}$ & $D_5\oplus (44)| \mathbb{Z}_4$ & Yes & 11 \\
\hline
${[312]}$ & $A_4\oplus A_1\oplus (230)|\mathbb{Z}_10$ & Yes & 23 \\
\hline
${[302]}$ & $A_5\oplus (4)$ & Yes & 24 \\
\hline
${[322]}$ & $D_5\oplus (20)|\mathbb{Z}_2$ & Yes & 20\\
\hline
${[303}$ & $A_2\oplus A_2\oplus A_1\oplus A_1$ & Yes & 36 \\
\hline
${[401]}$ &$A_5\oplus (24)|\mathbb{Z}_3$& Yes & 16\\
\hline
${[411]}$ & $A_4\oplus A_2$ & Yes & 15 \\
\hline
${[402]}$ & $A_3\oplus A_1\oplus A_1\oplus (8)|\mathbb{Z}_3$ & Yes & 32 \\
\hline
${[422]}$ & $A_3\oplus A_2\oplus (84)|\mathbb{Z}_6$ & Yes & 28 \\
\hline
${[501]}$ & $D_4\oplus A_1\oplus (10)|\mathbb{Z}_2$ & Yes & 20 \\
\hline
${[511]}$ & $A_5\oplus (114)|\mathbb{Z}_6$ & Yes  & 19\\
\hline
${[505]}$ & $A_4\oplus (4)\oplus (20)\mathbb{Z}_2$ & Yes & 100\\
\hline
${[555]}$ & $A_4\oplus (10)\oplus (6)|\mathbb{Z}_2$ & Yes & 75 \\
\hline
${[601]}$ & $A_3\oplus A_2\oplus A_1$ & Yes & 24\\
\hline
${[602]}$ & $A_3\oplus A_2\oplus (4)$ & Yes & 48\\
\hline
${[603]}$ & $A_3\oplus A_1^{\oplus 3}\oplus (12)|\mathbb{Z}_2$ & Yes & 72\\
\hline
${[606]}$ & $A_2\oplus A_2\oplus (4)\oplus(4)$ & Yes & 144\\
\hline
${[611}$ & $A_4\oplus A_1\oplus (230)|\mathbb{Z}_{10}$ & Yes & 23 \\
\hline
${[622]}$ & $A_2\oplus A_1^{\oplus 3}\oplus (66)|\mathbb{Z}_6$ & Yes & 44\\
\hline

\end{tabular}}}
\caption{This table shows the $T_0$ lattices computed by hand.}
\label{tab:T0}
 \end{center}
\end{table}
\clearpage
\begin{table}
\begin{center}
\setlength{\tabcolsep}{0.5em} 
\scalebox{0.7}{
\begin{tabular}{|c|c|c|c|}
\hline
$[abc]$ & $T_0$ & Contains Root? & $\det(T_{[abc]})=\det(T_0)$\\
\hline
\\[-1em]
$[333]$ & $\begin{psm}
\frac{1}{2}&\frac{1}{2}&-\frac{3}{2}&\frac{1}{2}&\frac{1}{2}&\frac{1}{2}&\frac{5}{2}&\frac{5}{2}\\
0&1&-1&0&0&0&1&1\\
0&0&0&1&0&0&2&1\\
0&0&0&0&1&0&2&1\\
0&0&0&0&0&1&2&3\\
0&0&0&0&0&0&3&1
\end{psm}$ & Yes & 27\\
\hline
\\[-0.75em]
$[404]$ & $\begin{psm}
\frac{1}{2}&\frac{1}{2}&\frac{1}{2}&\frac{1}{2}&-\frac{5}{2}&\frac{1}{2}&\frac{1}{2}&-\frac{9}{2}\\
0&1&0&0&-1&0&0&-6\\
0&0&1&0&-1&0&0&0\\
0&0&0&1&-1&0&0&0\\
0&0&0&0&0&1&0&-1\\
0&0&0&0&0&0&1&-1
\end{psm}$ & Yes & 64\\[3ex]
\hline
\\[-1em]
$[444]$ & $\begin{psm}
\frac{1}{2}&\frac{1}{2}&\frac{1}{2}&\frac{1}{2}&-\frac{5}{2}&\frac{1}{2}&\frac{1}{2}&\frac{11}{2}\\
0&1&0&0&-1&0&0&6\\
0&0&1&0&-1&0&0&2\\
0&0&0&1&-1&0&0&2\\
0&0&0&0&0&1&0&-1\\
0&0&0&0&0&0&1&-1
\end{psm}$ &Yes&48\\[5ex]
 \hline
\\[-1em]
$[666]$& $\begin{psm}
\frac{1}{2}&\frac{1}{2}&\frac{1}{2}&-\frac{5}{2}&\frac{1}{2}&\frac{1}{2}&\frac{1}{2}&\frac{3}{2}\\
0&1&0&-1&0&0&0&0\\
0&0&1&-1&0&0&0&0\\
0&0&0&0&1&2&0&-1\\
0&0&0&0&0&3&0&-1\\
0&0&0&0&0&0&1&1
\end{psm}$ & Yes&108\\[4ex]
\hline
\\[-1em]
$[777]$ & $\begin{psm}
\frac{1}{2}&\frac{1}{2}&-\frac{5}{2}&\frac{1}{2}&\frac{1}{2}&\frac{1}{2}&\frac{1}{2}&-\frac{1}{2}\\
0&1&-2&0&0&0&2&-1\\
0&0&0&1&0&0&4&-1\\
0&0&0&0&1&0&4&-1\\
0&0&0&0&0&1&4&-1\\
0&0&0&0&0&0&5&-1
\end{psm}$ & Yes&147\\[4ex]
\hline
\\[-1em]
$[888]$ & $\begin{psm}
\frac{1}{2}&\frac{1}{2}&\frac{1}{2}&\frac{1}{2}&-\frac{7}{2}&\frac{1}{2}&\frac{1}{2}&-\frac{3}{2}\\
0&1&0&0&-2&0&0&1\\
0&0&1&0&-1&0&0&4\\
0&0&0&1&-1&0&0&0\\
0&0&0&0&0&1&0&-1\\
0&0&0&0&0&0&1&-5
\end{psm}$ & Yes&192\\[4ex]
\hline
\end{tabular}}
\caption{This figure shows the generating matrices for all the $T_0$ lattice calculated using SAGE.}
\label{tab:T0SAGE}
\end{center}
\end{table}

\section{Niemeier lattices}

Niemeier showed that there are $24$, unique (up to isometry) even unimodular lattices of rank $24$ \cite{Niemeier}. An alternative proof that Niemeier's list is exhaustive was given by Venkov in \cite{Venkov} and amounts to showing first that there are 24 unique (up to isometry) root systems $\mathcal{N}(2)$ in dimension 24, and second that each of these gives rise to a unique even unimodular lattice $\mathcal{N}_{(\mathcal{N}(2))}$. In this section, the constructions of the Niemeier lattices from $ADE$ root lattices are given. This is necessary to find the primitive embeddings $T_0\xhookrightarrow{}\mathcal{N}$. 
\subsection{Gluing theory}\hfill

Witt's theorem guarantees that an integral lattice $L$ of rank $r$ whose generators have norm $2$ always has a decomposition of the form
\begin{equation}\label{adeglue}
L=\bigoplus_{i=0}^{r}A_i^{\oplus a_i}\oplus D_i^{\oplus d_i}\oplus E_i^{\oplus e_i}
\end{equation}
where $d_i=0\text{ for }i\leq 3$, $e_i =0\text{ for } i\neq 6,7,8$. We can often use these lattices to define a base lattice to which glue vectors are then added. Indeed, this is exactly how the construction of the Niemeier lattices proceeds. \\
\indent Gluing theory provides a way of decomposing a lattice into component lattices, but with more structure than that allowed in the direct sum.  Let the lattice $L$ be such that it has a sublattice composed of a direct sum of lattices $L_i$, $i\in \{1,\hdots,m\}$. In general, we have a lattice $\tilde{L}\defeq\bigoplus\limits_{i=1}^{m}L_i\subseteq L$. An arbitrary vector $v\in L$ will take the form:
\begin{equation}\label{eq:glue}
v=\sum_{i=1}^{m}v_i
\end{equation}
Every $v_i\in R_i$, where $R_i\defeq \spn_{\mathbb{R}}\{(b^{i})_1,\hdots,(b^{i})_{\rank{L_i}}\}$ with $b^{i}_{j}$ the $j$-th basis vector of $L_i$. That is, the $v_i$ are in the \textit{Euclidean} space whose basis is the lattice basis and are not, in general, in $L_i$. A vector $v_i\not\in l_i$, but $v_i\in L$ or $v=\vec{0}$ is called a \textit{glue vector}. Since $L$ is a lattice, we may add any $a_i\in l_i$ to $v_i$ to form another glue vector. As addition is the group operation for a lattice, the set of glue vectors obtained as $v_i + a_i$ form a coset of $L_i$ in $R_i$. A set of glue vectors containing exactly one vector from each orthogonal summand is called a \textit{transversal} of the cosets of $v_i$ in $L$. The glue vectors are often chosen so that they have minimal norm and the complete list of glue vectors for a lattice is a group that is termed the \textit{glue code}.
\begin{observation}
We observe that any glue vector $g_i$ in space $R_i$ is $g_i\in (L_i)^{*}$. This follows because the $L_i$ are mutually orthogonal, so we must have that $(a_i,L)=(a_i,L_i)$. $L$ is also integral, and so $(a_i,L)\in\mathbb{Z}$ - which is just the condition for $a_i\in (L_i)^{*}$.  
\end{observation}
There is a standard notation used for the glue vectors which we will employ in the sequel. This is given in Table \ref{gluetable}. 
\begin{table}
\begin{centering}
\setlength{\tabcolsep}{0.5em} 
{\renewcommand{\arraystretch}{1.5}
\begin{tabular}{|c|c|c|c|c|}
\hline
Root lattice $R$ & Glue vector notation $[a]$ & Glue vector & Norm & Glue group\\
\hline
\multirow{2}{*}{$A_n$} & $[0]$ & $\vec{0}$ & $0$ & \\
										& $[a]$ & $\begin{psm}(\frac{a}{n+1})^{n+1-a},-(\frac{n+1-a}{n+1})^{a}\end{psm}$ & $\frac{a(n+1-a))}{n+1}$ & $\mathbb{Z}_{n+1}$\\
\hline
 \multirow{4}{*}{$D_n$} & $[0]$ & $\vec{0}$ & $0$ & \\
 									& $[1]$ & $\begin{psm}(\frac{1}{2})^{n}\end{psm}$ & $\frac{n}{4}$ & $V_n\text{ if } n\in 2\mathbb{Z}$  \\
                                 & $[2]$ &  $\begin{psm}0^{n-1},1\end{psm}$ & $1$ & $\mathbb{Z}_n\text{ if } n\in 2\mathbb{Z}+1$ \\					
                                 & $[3]$ & $\begin{psm}(\frac{1}{2})^{n-1},-\frac{1}{2}\end{psm}$ & $\frac{n}{4}$ & \\\hline	
\multirow{3}{*}{$E_6$} & $[0]$ & $\vec{0}$ & $0$ & \\
										& $[1]$ & $\begin{psm}0;(-\frac{2}{3})^{2},\frac{1}{3};0\end{psm}$ & $\frac{4}{3}$ & $\mathbb{Z}_3$ \\
										& $[2]$ & $-[1]$ & $\frac{4}{3}$ & \\\hline
										
\multirow{2}{*}{$E_7$} & $[0]$ & $\vec{0}$ & $0$ & \\
										&  $[1]$ & $\begin{psm}(\frac{1}{4})^6,(-\frac{3}{4})^2\end{psm}$ & $\frac{3}{2}$ & $\mathbb{Z}_2$
\\\hline
$E_8$ & $[0]$ & $\vec{0}$ & $0$ & $\mathbb{Z}_1$\\
\hline
\end{tabular}}
\caption{In this table we summarise the standard notation used in \cite{CS} for the glue vectors. The glue group is isomorphic to the set of coset representatives of the quotients of the orthgonal sublattices $L_i$ with their dual $\frac{L_i^{*}}{L_i}$ - that is, it is ismorphic to the set of glue vectors not in $L_i$. }
\label{gluetable}
\end{centering}
\end{table}
\subsection{Niemeier lattices}\hfill

The construction of the Niemeier lattices is a well-known application of gluing theory. We begin with an orthogonal direct sum of $ADE$ root lattices as in equation \ref{adeglue}. There are $23$ such sums of rank $24$ and the empty root system \cite{Venkov}.Venkov then showed \cite{Venkov} that:
\begin{theorem}\label{thm:Ncomplete}
Every root system in 24 dimensions can be used to construct a unique even unimodular lattice by adding glue vectors to the root system.
\end{theorem}
Figuratively, we begin with an $ADE$ root lattice. Only one of these is unimodular, namely: $E_8^{\oplus 3}$. For a lattice to be unimodular, recall from the proof of Theorem \ref{thm: uni-s.d} that the fundamental region must be minimal. We can make the fundamental region smaller by adding in lattice points through the addition of glue vectors. Theorem \ref{thm:Ncomplete} tells us that all the even unimodular lattices (up to isometry) are constructed in this way. This is very similar to the notation we used in tables \ref{tab:T0} and \ref{tab:T0SAGE}, where we wrote $T_0$ as a root lattice $L$ and extra vectors that form an additive group $K$ in the quotient $\frac{T_0}{L}$. 
The complete list of glue vectors is given on \cite[p.~407]{CS} and we reproduce this as Table \ref{Niemeiertable}.
\clearpage
\begin{table}
\setlength{\tabcolsep}{0.5em} 
{\renewcommand{\arraystretch}{1.5}
\begin{tabular}{|c|c|c|c|}
\hline
Root system $R$ & Glue code generators & Number of glue vectors\footnote{Note that we always have the case identity vector as a glue vector} & Number of roots\\
\hline
$D_{24}$ & $[1]$ & 2 & 1104\\
\hline
$E_8^{\oplus 3}$ & $[000]$ & 1 & 720\\
\hline
$D_{16}\oplus E_8$ & $[10]$ & 2 & 720\\
\hline
$A_{24}$ & $[5]$ & 5 & 600\\
\hline
$D_{12}\oplus D_{12}$ & $[(12)]$ & 4 & 628 \\
\hline
$A_{17}\oplus E_7$ & $[31]$ & 6 & 432 \\
\hline
$D_{10}\oplus E_7\oplus E_7$ & $[110],[301]$ & 4 & 432 \\
\hline
$A_{15}\oplus D_9$ & $[21]$ & 8 & 384 \\
\hline
$D_8^{\oplus 3}$ & $[(122)]$ & 8 & 336 \\
\hline
$A_{12}\oplus A_{12}$ & $[15]$ & 13 & 312 \\
\hline
$A_{11}\oplus D_7\oplus E_6$ & $[111]$ & 12 & 288 \\
\hline
$E_6^{\oplus 4}$ &$[1(012)]$& 9 & 288  \\
\hline
$A_9\oplus A_9\oplus D_6 $ & $[240],[501],[053]$ & 20& 240\\
\hline
$(D_6)^{\oplus 4}$ & $A[0123]$ & 16 & 240\\
\hline
$A_8^{\oplus 3}$ & $[(114)]$ & 27 & 216 \\
\hline
$A_7^{\oplus 2}\oplus D_5^{\oplus 2}$ & $[1112],[1721]$ & 32 & 192 \\
\hline
$A_6^{\oplus 4}$ & $[1(216)]$ & 49 & 168 \\
\hline
$A_5^{\oplus 4}\oplus D_4$ & $[2(024)0],[33001],[30302],[30033]$ & 72 & 144 \\
\hline
$D_4^{\oplus 6}$ & $[111111],[0(02332)]$ & 64 & 144 \\
\hline
$A_4^{\oplus 6}$ & $[1(01441)]$ & 125 & 120\\
\hline
$A_3^{\oplus 8}$ & $[3(2001011)]$ & 256 & 96 \\
\hline
$A_2^{\oplus 12}$ & $[2(11211122212)]$ & 729 & 72 \\
\hline
$A_1^{\oplus 24}$ & $[1(00000101001100110101111)]$ & 4096 & 48 \\
\hline
$\Lambda_{24}$ & - & - & None \\
\hline
\end{tabular}}
\caption{This table gives the glue vectors that must be added to the 23 root systems of rank 24 to obtain the Niemeier lattices. The Leech lattice is included for completeness, and $A[\hdots]$ indicates that the glue code is formed by the action of the alternating group on the glue vectors.}
\label{Niemeiertable}
\end{table}
\section{Sphere packing}
The sphere packing problem is the second part of  Hilbert's eighteenth problem \cite{Hilbert}, and is still one of the famous open problems in mathematics. The modern version of the problem is the following:
\begin{qst}
In an arbitrary dimension $n$, which arrangement of spheres of fixed radius in $R^{n}$ has the highest possible number of spheres per unit volume?
\end{qst}
Many variations of the problem have been considered: the spheres may be allowed to have different radii \cite{differentradius}; a variety of different geometric objects may be used \cite{otherpacking} and different ambient spaces to $\mathbb{R}^{n}$, most notably hyperbolic spaces, have been considered \cite{hyperbolicpacking}. Our immediate focus is a weaker form of the standard problem\footnote{An excellent reference for the sphere packing problem and it's applications to communication theory and physics is \cite{Cohn}.}. \\
\indent What arrangements of spheres are possible? A natural approach is to place sphere centers on the points of a lattice. Such an arrangement is called a \textit{lattice sphere packing}. A lattice packing imposes structure on the sphere packing, and so it is unsuprising that there exists denser sphere packings than lattice packings in certain dimensions. More generally we can consider:
\begin{definition}[Periodic packing]
Let $\mathcal{C}$ be a set whose entries are vectors in a Euclidean space $\mathbb{R}^{n}$ from the origin to centers of spheres of a given radius. $\mathcal{C}$ is called a \textit{periodic} packing if there is some collection of lattices $L_{1},\cdots,L_{k}$ such that $\mathcal{C}=\bigcup\limits_{i=1}^{k}L_{i}$.
\end{definition}
That we need not weaken further is the content of the following theorem:
\begin{theorem}
The densest periodic sphere packing on a region approaches the densest possible packing as the region becomes infinite. 
\end{theorem}
We provide a sketch due to \cite{Cohn} - a formal proof can be found in reference (37) of that work, provided here for convenience as \cite{Cohn37}. Consider an optimal sphere packing in a given dimension and take a large `tile' of that packing. The sphere packing in that tile can be repeated to cover the space. Such a packing would be periodic and only at the edges of the tile could the spheres in the periodic packing possibly overlap. These can always be removed, and by taking the tile to be large enough, we see that a periodic packing is arbitrarily close to the optimal packing. \\

The \textit{density of a sphere packing}, $\rho_S$, is simply the proportion of volume of a space contained within the spheres. In a lattice periodic packing there is exactly one sphere per fundamental region and so the density for spheres of radius $R$ is given by:
\begin{definition}[Density of a sphere packing]
\begin{equation}
\rho_S\defeq\frac{\mathrm{volume \text{ }of\text{ }sphere}}{\mathrm{volume\text{ }of\text{ }fundamental\text{ }region}}=\frac{(\pi)^{\frac{n}{2}}}{\Gamma(\frac{n}{2}+1)}R^{n}\frac{1}{\sqrt{\det(L)}}
\end{equation}
\end{definition}
For a given dimension, the volume prefactor is a constant, and so a dividing by the volume of the $n$-dimensional unit sphere, gives a more transparent measure - the \textit{center density}, $\delta_L$
\begin{definition}
\begin{equation}
\delta_L=\frac{(R)^{\frac{n}{2}}}{4}\frac{1}{\sqrt{\det(L)}}=\frac{(\lambda_1)^{\frac{n}{2}}}{4}\frac{1}{\sqrt{\det(L)}}
\end{equation}
\end{definition}
Where in the last equality, $\lambda_1$ is the shortest length vector in the lattice. 
\subsection{Laminated lattices}\hfill

We want to make a connection between lattice sphere packgings and the possible frame lattices $W$. To do so we need to introduce the best known, and sometimes best possible, lattice sphere packings. We require the notion of a \textit{laminated lattice}:
\begin{definition}[Laminated lattice]
A \textit{laminated lattice} is one constructed as follows: Begin with the lattice consisting of a single point $\Lambda_0$. Then, from all one-dimensional lattices of minimal norm 4 that have the one-point lattice as a sublattice choose those with the smallest determinant. We then repeat this procedure to form $\Lambda_n$ from $\Lambda_{n+1}$.
\end{definition}
We denote the laminated lattices as $\Lambda_n$, and they form a sequence with all $\Lambda_n$ being a primtive sublattice of $\Lambda_{n+1}$. The Leech lattice is the unique laminated lattice in dimension 24, and so all of the laminated lattices appearing in the Table \ref{tab:Packings} appear as sublattices of the Leech. There is another sequence of lattices which contain some of the densest known lattice sphere packings. These are the $K_{n}$ lattices, which are \textit{also} primitive sublattices of the Leech lattice \cite[Chapter~6]{CS}. 
We have the result that: (\cite[Corollary~7,~Chapter~7]{CS})
\begin{theorem}
Each of the laminated lattices $\Lambda_0,\cdots,\Lambda_8$ and $\Lambda_{16},\cdots,\Lambda_{23}$ are the densest primitive sublattices of the Leech lattice in their dimension. Moreover, each of these lattice sphere packings is the densest \textbf{known} lattice packings in their dimension.
\end{theorem}
Further:
\begin{theorem}
$\Lambda_0,\cdots,\Lambda_8$, are the \textit{unique} best possible lattice packings in their dimension. 
\end{theorem}
These facts are summarised in the Table \ref{tab:Packings}.\\
\indent Moreover, we can see that the Leech lattice contains \textit{every} densest known lattice in dimension $n\leq 24$ as a primitive sublattice. Looking at Table \ref{tab:Packings}, we find a striking coincidence:
\begin{observation}
The densest known lattice packing in 18 dimensions has a determinant exactly equal to that required by the F-theory tadpole cancellation condition. The author does not know of any reason for this to be so, and this fact appears as an intriguing connection between mathematics and physics that wants explanation.
\end{observation}
\pagebreak
\begin{table}
\begin{center}
\setlength{\tabcolsep}{0.5em} 
{\renewcommand{\arraystretch}{1.5}
\begin{tabular}{|c|c|c|c|}
\hline
$Lattice$ &determinant&Optimality proof?&Uniqueness proof? \\
\hline
$\Lambda_0$ & - & Yes & Yes \\
\hline
$A_1'$ & $4$ & Yes & Yes \\
\hline
$A_2'$ & $12$ & Yes & Yes \\
\hline
$A_3'$ & $32$ & Yes & Yes \\
\hline
$D_4'$ & $64$ & Yes & Yes\\
\hline
$D_5'$ & $128$ & Yes & Yes\\
\hline
$E_6'$ & $192$ & Yes & Yes\\
\hline
$E_7'$ & $256$ & Yes & Yes\\
\hline
$E_8'$ & $256$ & Yes & Yes \\
\hline
$\Lambda_9$ & $512$ & No & No\\
\hline
$\Lambda_{10}$ & $768$ & No & No\\
\hline
$K_{11}$ & $972$ & No & No\\
\hline
$K_{12}$ & $729$ & No & No\\
\hline
$K_{13}$ & $972$ & No & No\\
\hline
$\Lambda_{14}$ & $768$ & No & No\\
\hline
$\Lambda_{15}$ & $512$ & No & No\\
\hline
$\Lambda_{16}$ & $256$ & No & No\\
\hline
$\Lambda_{17}$ & $256$ & No & No\\
\hline
$\Lambda_{18}$ & $192$ & No & No \\
\hline
$\Lambda_{19}$ & $128$ & No & No\\
\hline
$\Lambda_{20}$ & $64$ & No & No\\
\hline
$\Lambda_{21}$ & $32$ & No & No\\
\hline
$\Lambda_{22}$ & $12$ & No & No\\
\hline
$\Lambda_{23}$ & $4$ & No & No\\
\hline
$\Lambda_{24}$ & $1$ & Yes & Yes\\
\hline
\end{tabular}}
\caption{This table shows some of the best known lattice sphere packings. The third coloumn indicates whether there is a proof that the lattice packing shown is the best possible. A sphere packing with spheres of radius 1 has been assumed, and the prime indicates that the lattices have been scaled from their conventional form.}
\label{tab:Packings}
\end{center}
\end{table}
\section{Leech Lattice}
The optimality of the Leech is such that it warrants it's own section and provides a good illustration of \textit{coding theory}, which has been crucial in constructing dense lattice packings.
Conway has given an elegant proof of the fact that: \cite[Chapter~12]{CS}
\begin{theorem}\label{thm:Leechunique}
Up to isometry, there is a unique even unimodular lattice in $24$ dimensions, $\Lambda_{24}$, with an empty root system. 
\end{theorem}
The Leech lattice has a number of remarkable properties: 
\begin{enumerate}
	\item It was proven in 2016 \cite{Leechunique} to be the densest periodic sphere packing possible in dimension $24$. 
	\item There are exactly 23 types of inequivalent regions of $R^{n}$ maximally distant from any lattice point, the \textit{deep holes} of the lattice. The norm of the difference between the lattice vectors forming the vertices of the hole is one of $4,6\text{ or }8$ and so we can construct a Dynkin diagram from the hole. Each of the 23 Dynkin diagrams can be associated to glue vectors to reconstruct exactly one Niemeier lattice. 
	\item The best known lattice sphere packing in every dimension $\leq 24$ is a primitive sublattice of the Leech lattice.\cite[Chapter~6]{CS}
	\end{enumerate}
	
Here we give two constructions of $\Lambda_{24}$. The first has the advantage of being explicit and proceeds from the \textit{extended binary Golay Code}. The second has the advantages of being brief and that it can be used to give a uniform proof of point $(2)$ above. \cite{LeechBOR}

\subsection{Coding theory}\hfill

Coding theory is the most commonly used tool to construct dense lattice packings, and is also a common source of applications of lattice theory. We will largely follow \cite{LandCWE} in providing an overview. 
Classically, information is transported as a finite string of symbols. The string is a \textit{code} and there is an intimate connection between coding theory and lattice theory. A code can be used to construct a lattice and a lattice also defines a code. Formally:
\begin{definition}[Code]
A \textit{code} $C$ of length $n$ is any non-empty proper subset of the finite field $\mathbb{F}_{q}^{n}$. 
\end{definition}
The prime number $q$ is the number of values a single entry in the code can take. The binary code standard in digital communication would have $q=2$, but ternary ( with $q=3$) and other codes have also been used to construct lattices. \\
\indent We shall use a \textit{binary linear code} to construct the Leech lattice:
\begin{definition}
A \textit{linear code} $C_L$ is a subspace of $\mathbb{F}_{q}^{n}$ such that $f_1,f_2\in C_L\implies f_1+f_2\in C_L$.
\end{definition}
An important concept in coding theory is the \textit{weight} of a code:
\begin{definition}[weight]
The \textit{weight} of a code $c$, $w(c)$, is the number of non-zero entries in the code. 
\end{definition}
This can be used to give a notion of \textit{distance} between two codes: 
\begin{definition}[Hamming distance]
The \textit{Hamming distance} between two codes $c_1$ and $c_2$ is $w(c_1-c_2)$.
\end{definition}
This gives a measure of how easy it is to distinguish the two codes. Classically, a communication error in a binary code can be modelled as a \textit{bit-flip} which takes $x_i\rightarrow x_i+1\mod 2$. Thus the distance can be thought of as the extent to which an error is liable to produce the code $c_1$ from $c_2$ and conversely. 
We have the following notation:
\begin{notation}
A linear $[n,k,d]$ code is one with $k$ codewords (strings of symbols, or blocks)of length $n$ with a minimum Hamming distance $d$.
\end{notation}
Given a bilinear form $(\cdot,\cdot)$ on the field $\mathbb{F}_{q}^{n}$ of a code $C$, we can define it's \textit{dual code} $C^{\perp}$ as:
\begin{definition}[Dual code]
\begin{equation}
C^{\perp}\defeq\{c^{*}\in\mathbb{F}_{q}^{n} | (c^{*},c)=0\text{ }\forall c\in C \}
\end{equation}
\end{definition}
\subsubsection{Lattices from codes}
There are a number of different constructions from a code to a lattice  (\cite[Chapters~5,7]{CS}). Our construction will use the following scheme. Let $L_{r}$ be a rank $r$ lattice and consider the map $\rho : L_{r}\rightarrow \frac{L_{r}}{2L_{r}}\cong\mathbb{F}_{2}^{n}$. Conversely, $\rho^{-1}:\mathbb{F}_{2}^{r}\rightarrow L_{r}$.\\
\indent We denote the lattice defined by a code $C$ as $L_C$. Explicitly:
\begin{definition}
For a binary $[n,k,d]$ code $C$, we can define the rank $n$ lattice:
\begin{equation}
L_C\defeq \frac{1}{\sqrt{2}}\rho^{-1}(C)
\end{equation}
with:
\begin{equation}
l\in L_C=\frac{1}{\sqrt{2}}(c+2z),\text{ }z\in\mathbb{Z}^{n},c\in C
\end{equation}
\end{definition}
We have the following useful theorems relating the linear code $[n,k,d]$ to the lattice which may be constructed from it \cite[Chapter~7]{CS}:
\begin{theorem}
$L_C$ is a lattice $\iff$ $C$ is a linear code.
\end{theorem}
\begin{theorem}
$(L_C)^{*}=L_{C^{*}}$. In particular $L_C=(L_C)^{*}\iff C=C^{\perp}$
\end{theorem}
\begin{theorem}
$L_C$ is integral $\iff$ $C\subseteq C^{*}$
\end{theorem}
\begin{theorem}
\begin{equation}
\det(L_C)=2^{n-2k}
\end{equation}
\end{theorem}
\begin{theorem}
$L_C$ is even $\iff$ $\frac{wt(c)}{4}\in\mathbb{Z}\text{ }\forall c\in C$ ($C$ is \textit{doubly even}).
\end{theorem}
\subsection{Constructing the Leech lattice}\hfill

The Leech lattice is even unimodular and of rank 24, so we will need a self-dual doubly even code with $n=24$. We have (theorem 2.6 of \cite{LandCWE}):
\begin{theorem}
Let $C$ be a doubly even binary self-dual $[24,12,8]$ code containing 0. Up to equivalence, there is at most one such code. 
\end{theorem}
This code is the exteneded binary Golay code, conventionally $\mathcal{G}_{24}$. As before this allows us to construct the lattice:
$$L_C=\frac{1}{\sqrt{2}}\rho^{-1}(\mathcal{G}_{24})$$
whose elements take the form $l=\frac{1}{\sqrt{2}}(c+2z)$
Since the Golay code is doubly even the sum over component of $l$, $\sum\limits_{i=1}^{24}l_i\in 2\mathbb{Z}$. This allows us to define a lattice $A$ consisting of all elements whose components sum to a multiple of 4:
\begin{equation}
A\defeq \{a\in L_{\mathcal{G}_{24}} | \sum\limits_{i=0}^{24}a_i\in 4\mathbb{Z}\}
\end{equation}
We can also define:
\begin{equation}
B\defeq \{b\in  L_{\mathcal{G}_{24}} | \sum\limits_{i=0}^{24}b_i\not\in 4\mathbb{Z}\}
\end{equation}
We can then introduce the lattice $D\defeq A\sqcup (\frac{1}{2}\vec{1}+B)$. It can then be shwon that \cite[Chapter~4]{LandCWE}
\begin{theorem}
$D$ is an even unimodular lattice with no roots.
\end{theorem}
From Theorem \ref{thm:Leechunique} it follows that $D$ is the Leech lattice, as required. 
\subsection{Construction from the even Lorentzian lattice $\Lambda^{(25,1)}$}\hfill

Conway showed \cite[Chapter~26]{CS} that a particularly simple construction of the Leech lattice is obtained by considering the unique even unimodular Lorentzian lattice (that is one in which the ambient space $\mathbb{R}^{(n,1)}$ is equipped with a Minkowski metric) of rank 26, $\Lambda^{(25,1)}$. \\
\indent For $x\in\mathbb{R}^{(n,1)}$ we denote by $(x_1,\hdots,x_n | x_{n+1})$ the vector with $x^{2}=x_{1}^{2}+\cdots+x_{n}^{2}-x_{n+1}^{2}$. An \textit{isotropic vector} $v$ is one with $v^{2}=0$. One such isotropic vector is $v=(1,2,3,\cdots,24|70)$. This vector can be used to show that:
\begin{theorem}\label{thm:Leechhyperbolic}
\begin{equation}
\Lambda_{24}\cong \frac{w^{\perp}\subset L^{(25,1)}}{\mathrm{span}_{\mathbb{Z}}\{w\}}
\end{equation}
\end{theorem}
Where we abuse notation slightly to allow the isomorphism symbol to indicate isometry. \\
\subsubsection{Finding the Niemeier lattices}
\indent Consider the set $\{f_1,\cdots,f_n\}$ where the $f_i$ are the vectors generating a root system of a Niemeier lattice and  $\{g_1,\cdots,g_{N_g}\}$ is a set of glue vectors. Then it can then be shown from Theorem \ref{thm:Leechhyperbolic} ( \cite{LeechBOR}, \cite[Chapter~23]{CS}\footnote{A good pedagogical reference is \cite[Chapter~4]{LandCWE}}) that the construction:
\begin{equation}
\sum_{i=1}^{24}a_if_i+\sum_{i=1}^{N_g}b_ig_i\text{ }\forall \{m_1,\cdots,m_{24}\},\{g_1,\cdots,g_{N_g}\}\text{ such that }\sum_{i=1}^{24}b_i=0
\end{equation}
forms a lattice isometric to the Niemeier lattice, and:
\begin{equation}
\sum_{i=1}^{24}a_if_i+\sum_{i=1}^{N_g}b_ig_i\text{ }\forall \{m_1,\cdots,m_{24}\},\{g_1,\cdots,g_{N_g}\}\text{ such that }\sum_{i=1}^{24}b_i=0+\sum_{i=1}^{N_g}b_i=0
\end{equation}
forms a lattice isometric to the Leech lattice. \\

\section{Strings and Sphere packings}
To answer the question posed in this disseration, the Kneser-Nishiyama method can be carried out explicitly for the 34 transcendental lattices listed in tables \ref{tab:T0} and \ref{tab:T0SAGE}. This was done in \cite{1312} for embeddings of $T_{[302]}$ into some of the Niemeier lattices. This calculation is technically difficult since, although only one set of possible $T_0$ lattices is required, \textit{every} primitive embedding into \textit{every} Niemeier lattice must be found. Here we take a different approach, and in so doing make a previously unknown connection between string theory and the sphere packing problem. 
\subsection{The frame lattice}\hfill

\indent Recall that the formulation of our question in lattice theoretic terms is:
\begin{qst}
Does every frame lattice $W\defeq T_0^{\perp}\subset\mathcal{N}$ contain a root?
\end{qst}
What do we know about $W$? The Kneser-Nishiyama method constructs the Niemeier lattice exactly as in Proposition \ref{Nik1}. In particular, this entails:
\begin{equation}
\mathcal{N}_i=(W\oplus T_0);\Delta
\end{equation}
In turn, this allows us to use Proposition \ref{Nik2} to deduce that:
\begin{fact}
The discriminant group of $W$ is isomorphic to the discriminant group of $T_0$. Exactly the same argument establishes that $G_{T_{[abc]}}\cong G_{T_0}$.
\end{fact}
Which implies:
\begin{corollary}
$$\det(W)=\det(T_0)\iff\mathrm{vol_{FR}(W)}=\mathrm{vol_{FR}}(T_0)$$ 
and:
$$\det(T_0)=\det(T_{[abc]})\iff\mathrm{vol_{FR}}(T_{[abc]})=\mathrm{vol_{FR}}(T_{0})$$
\end{corollary}
\subsection{Bounds}\hfill

\begin{strategy}
We attempt a proof by contradicition, assuming that $W$ has no roots.
\end{strategy}
If $W$ has no roots, we would expect it's fundamental region to be larger than if it had, since there are less vectors of smaller size. So we attempt to show that the minimal possible fundamental region of $W$ is larger than the largest allowed fundamental region of $T_{[abc]}$. \\
\indent Recall that the center density for a lattice $L$ is given by:
\begin{equation}
\delta_L=\frac{(\lambda_1)^{\frac{n}{2}}}{4}\frac{1}{\sqrt{\det(L)}}
\end{equation}
\begin{remark}\label{thm: scaleinv}
The center density is scale invariant, since scaling every lattice vector by some constant $c$ scales the radius as $c^{\frac{n}{2}}$, but also scales the determinant as $c^{n}$.
\end{remark}
Where $\lambda_1$ is the minimum norm vector in $L$ and the $\frac{1}{4}$ appears because if the minimal norm is 4, then spheres of radius $1$ can be packed on the lattice. We are assuming that $W$ has no roots, so the minimal length vector possible in $W$ is 4. A lattice sphere packing has exactly one sphere per fundamental region. Fixing the radius of the spheres, the densest sphere packing is achieved by minimising the volume of the fundamental region, which may be considered as the volume required for the lattice to contain one sphere. Thus, we see that finding the minimum fundamental region whilst fixing the minimal length vector corresponds exactly to the sphere packing problem.\\
\indent This allows us to use the results of that field, which we earlier summarised. There is a literature dating from Rogers' bound \cite{Roger} that attempts to estabish upper bounds on the center density. The most recent is the work of Cohn, \cite{Cohnbounds} who used linear programming to show that:
\begin{theorem}\label{thm: bound}
The center density in 18 dimensions is less than or equal to 0.16503\footnote{Note that rounding errors do not affect our discussion so we omit them for clarity.}. 
\end{theorem}
This allows us to show:
\begin{corollary}\label{thm: mn4}
Every lattice $W$ that we consider can have minimal norm of at most 4.
\end{corollary}
\begin{proof}
This is a straightforward application of Theorem \ref{thm: bound}. Inverting the forumla for the center density we have:
\begin{equation}
\det(L)=\bigg(\frac{(\frac{\lambda_1}{4})^{\frac{n}{2}}}{\delta_L}\bigg)^{2}
\end{equation}
So an upper bound on $\delta_L$ establishes a lower bound on $\det(L)$. Thus, setting $\lambda_1=4$, we have:
\begin{equation}\label{eq:bound}
\det(L_{\lambda_1=4})\geq 36.176
\end{equation}
Setting $\lambda_1=6$ we have:
\begin{equation}
\det(L_{\lambda_1=6})\geq 54264.6
\end{equation}
which is larger than the largest value of $\det(T_{[abc]})$ in tables \ref{tab:T0} and \ref{tab:T0SAGE}. Hence, the minimal norm vector in $W$ is at most 4.
\end{proof}
\begin{corollary}
There are only 13 $T_{[abc]}$ for which it is possible that $W$ does not contain a root.
\end{corollary}
\begin{proof}
We simply eliminate all those $T_{[abc]}$ with a fundamental region smaller than the minimum allowed value given by equation \ref{eq:bound}. Checking tables \ref{tab:T0} and \ref{tab:T0SAGE} we see that there are only 13 such with a determinant greater than or equal to the minimum allowed value.
\end{proof}
\begin{observation}
If any frame lattice $W$ does not have a root, then it is as least as good a lattice sphere packing as the best known lattice packing, $\Lambda_{18}$.
\end{observation}
\begin{proof}
Let $W$ be a primitively embedded 18 dimensional frame lattice which is also a sublattice of $\mathcal{N}_i$.  The highest possible determinant of $W$ is 192 from tables \ref{tab:T0} and \ref{tab:T0SAGE}. Further, Corollary \ref{thm: mn4} requires that the minimal norm of $W$ be 4. Thus, $\det(W)\leq 192\implies\delta_W\geq\delta_{\Lambda_{18}}$.
\end{proof}
We can also see that:
\begin{observation}
The densest eighteen dimensional primitive sublattice of the Leech lattice with vectors of minimal norm 4, has a discriminant group isomorphic to $E6'$. 
\end{observation}
\begin{proof}
This follows from Proposition \ref{Nik2}. A primitively embedded sublattice of the Leech will have a six-dimensional lattice as orthogonal complement. $E6'$ is the unique densest lattice in this dimension and so Proposition \ref{Nik2} requires that they have isomorphic discriminant groups. 
\end{proof}
From Proposition \ref{Nik2}, we know that if $\Lambda_{18}$ is to be found as a frame lattice $W$ it must be the case that $G_{W}\cong G_{T_{[888]}}$, since this is the only lattice with the required determinant. 
\begin{proposition}
There is no $W$ that is isometric to $\Lambda_{18}$.
\end{proposition}
\begin{proof}
This is a simple check of the discriminant groups. Using code written by Braun, we find that $G_{{\Lambda_{18}}}=(\mathbb{Z}_2)^{\times 5}\times\mathbb{Z}_6=(\mathbb{Z}_2)^{\times 6}\times\mathbb{Z}_3$ and $G_{T_{[888]}}=\mathbb{Z}_8\times\mathbb{Z}_{24}=\mathbb{Z}_{8}^{\times 2}\times\mathbb{Z}_3$. But: $\mathbb{Z}_{8}^{\times 2}\times\mathbb{Z}_3\neq\mathbb{Z}_2^{\times 6}\times\mathbb{Z}_3$. So there can be no $T_0$ such that $\Lambda_{18}=(T_0^{\perp}\subset\mathcal{N}_i)$.
\end{proof}
To proceed further, we require the following theorem \cite{NortonNtoL}
\begin{theorem}\label{Nsqrt2}
For any Niemeier lattice $\mathcal{N}$, the lattice that results from scaling every element by $\sqrt{2}$, $\sqrt{2}\mathcal{N}$, has at least one (and in general several) isometric embedding into the Leech lattice. 
\end{theorem}
This allows us to show:
\begin{corollary}
If any frame lattice $W$ appearing in some $\mathcal{N}_i$ has no roots and has $\det(W)\leq 192$, then there is an isometric lattice $W'\defeq\sqrt{2}W$ which is a sublattice of the Leech lattice which is at least as good a sphere packing as $\Lambda_{18}$.
\end{corollary}
\begin{proof}
By Theorem \ref{Nsqrt2}, we have a primitive embedding $W'\xhookrightarrow{}\Lambda_{24}$, and by Remark \ref{thm: scaleinv}, $\delta_W=\delta_{W'}$. $W'$ is then isometric to a sublattice of the Leech lattice, $\widetilde{W}$. Note that $\widetilde{W}$ is in general \textit{not primitive}. We have that $\delta_{\widetilde{W}}\geq\delta_{\Lambda_{18}}$, as required. 
\end{proof}

The above theorems make plausible a conjecture that answers the question posed in this dissertation in the affirmative:
\begin{conjecture}
Every frame lattice $W$ corresponding to a lattice $T_{[abc]}$ contains a root. 
\end{conjecture}
How could this conjecture prove to be incorrect? One of the following would have to be true:
\begin{enumerate}
\item There is a lattice sphere packing in dimension 18 contained in a Niemeier lattice that is at least as dense as the laminated lattice $\Lambda_{18}$ and is not isometric to it. This would mean that the Leech lattice does not contain all densest lattices as primitively embedded sublattices in dimension $< 24$.
\item There is a (non-primitively embedded) sublattice of the Leech lattice with a higher density than the densest primitive sublattice of the Leech lattice.
\end{enumerate}
Either of these would entail that there is a previously unknown lattice sphere packing that does not fit into the previous pattern of the best known lattice sphere packings up to dimension 24. \\
\indent Any of these results would be mathematically interesting. It is, however, implausible that either is true. The Leech is stirkingly optimal as a sphere packing in 24 dimensions. It is unlikely that the Niemeier lattices could contain a denser lattice sphere packing. If this was so, then the densest lattice sphere packing would be a non-primitively embedded sublattice of the Leech lattice and we do not know of any dense sphere packings\footnote{ Up to scaling - we can always scale a primitively embedded sublattice without changing it's center density to generate a counterexample otherwise.} that occur in a similar way. 
\section{Conclusion}
\indent The classification of elliptic fibrations on $K3$ surfaces through lattice embeddings has been our main focus here. However, the utility of lattice theory in physics begins with the simple way in which root systems can be used to construct Lie algebras through the Serre construction. Moreover, the $ADE$ singularities are essential to the physical interpretation of our result, and their classification is given in terms of Dynkin diagrams. The presentation here has been necessarily mathematical, but we have endeavoured to avoid formalism unless necessary. It is thus hoped that this work may also serve as an accesible introduction to group lattice theory in physics for those without a background in abstract algebra or modern group theory. \\
\indent It is still very difficult to say if a given string theory model contains the standard model. We can, however, ask a more modest question: Do string theories contain non-abelian gauge symmetry, and if so how typical is this? The primary purpose of this dissertation was to answer this question in a promising and tractable limit of M-theory: F-theory with a Calabi-Yau four-fold $CY_4=K3\times K3$ and Picard number 20. We followed \cite{1401} in using the Kneser-Nishiyama method to transform this into a problem of lattice embeddings. Upon doing so, our question became: do all frame lattices $W$ contain a root?\\
\indent Presenting the necessary lattice theory and arguing by contradiction we were able to demonstrate a previously unknown connection with the sphere packing problem. Remarkably, the densest known lattice packing in 18 dimensions has a determinant that is exactly the largest determinant of the transcendental lattice allowed by the F-theory Tadpole Cancellation Condition. We further proved that 21 of the 34 possible cases must lead to a frame lattice with roots. Previously, this was only partially known for just one of the transcendental lattices, $T_{[302]}$. Further, we have for the first time computed all of the necessary $T_0$ lattices to carry out this calculation explicitly. It appears, however,  to yield more insight to pursue the alternative method given here. It is the view of the author that extending the methods in this paper to cover \textit{all} cases should not prove too difficult.\\
\indent Although progress has been made, our conjecture remains unproven. It has further been shown, however, that the falsity of the conjecture would itself consistute an intriguing result in the sphere packing problem. If the conjecture is true, then we can begin to ask further questions: what kinds of gauge symmetry occur in the F-theory models we have been considering? Can we show that other F-theories also generically exhibit non-abelian gauge symmetry? Do any such models contain the gauge symmetry of the standard model?
\clearpage

\section{Supplementary Information}
\begin{figure}[h]
\begin{lstlisting}
T_333 = span([[2,1,1,0,0,0,0,0],[3/2,1/2,-1/2,1/2,1/2,1/2,1/2,-3/2]],ZZ)
 #Defining the transcendental lattice
T_333Q = span([[2,1,1,0,0,0,0,0],[3/2,1/2,-1/2,1/2,1/2,1/2,1/2,-3/2]],QQ)
 #Taking the tensor product over Q
T_333Qperp = T_333Q.complement().intersection(E8) #Calculating T_0
print(T_333Qperp.gram_matrix().det()==6*6-9) #Determinant Check


T_404=span([[2,1,1,1,1,0,0,0],[1/2,-5/2,1/2,1/2,1/2,-1/2,-1/2,-1/2]],ZZ)
T_404Q =span([[2,1,1,1,1,0,0,0],[1/2,-5/2,1/2,1/2,1/2,-1/2,-1/2,-1/2]],QQ)
T_404perp=T_404Q.complement().intersection(E8)
T_404perp.gram_matrix().det()==64

T_444 =span([[2,1,1,1,1,0,0,0],[1/2,5/2,1/2,1/2,-1/2,-1/2,-1/2,-1/2]],ZZ)
T_444Q = span([[2,1,1,1,1,0,0,0],[1/2,5/2,1/2,1/2,-1/2,-1/2,-1/2,-1/2]],QQ)
T_444perp = T_444Q.complement().intersection(E8)
T_444perp.gram_matrix().det()==48


T_603 = span([[2,2,1,1,1,1,0,0],[3/2,-3/2,1/2,1/2,-1/2,-1/2,-1/2,1/2]],ZZ)
T_603Q = span([[2,2,1,1,1,1,0,0],[3/2,-3/2,1/2,1/2,-1/2,-1/2,-1/2,1/2]],QQ)
T_603perp= T_603Q.complement().intersection(E8)
print(T_603.gram_matrix().det())==72


T_666 = span([[3,1,1,1,0,0,0,0],[5/2,-1/2,-1/2,-1/2,-1/2,-1/2,3/2,-3/2]],ZZ)
T_666Q = span([[3,1,1,1,0,0,0,0],[5/2,-1/2,-1/2,-1/2,-1/2,-1/2,3/2,-3/2]],QQ)
T_666Qperp = T_666Q.complement().intersection(E8)
print(T_666Qperp.gram_matrix().det()==108


T_777 = span([[3,2,1,0,0,0,0,0],[1/2,5/2,1/2,1/2,1/2,1/2,1/2,5/2]],ZZ)
T_777Q = span([[3,2,1,0,0,0,0,0],[1/2,5/2,1/2,1/2,1/2,1/2,1/2,5/2]],QQ)
T_777Qperp = T_777Q.complement().intersection(E8)
print(T_777Qperp.gram_matrix().det()==147

T_888=span([[3,2,1,1,1,0,0,0],[5/2,1/2,-3/2,1/2,1/2,1/2,5/2,1/2]],ZZ)
T_888Q =span([[3,2,1,1,1,0,0,0],[5/2,1/2,-3/2,1/2,1/2,1/2,5/2,1/2]],QQ)
T_888perp=T_888Q.complement().intersection(E8)
print(T_888perp.gram_matrix().det()==192



\end{lstlisting}
\caption{SAGE code used to calculate the $T_0$ lattices from the transcendental lattices. Since the first line defining the primitive embedding is non-unique, it is only this choice of primitive embedding that will generate this particular set of $T_0$'s, but \textit{whatever} primitive embeddings is chosen, the $W$'s so found are the same. }
\label{fig:SAGECODE}
\end{figure}

\begin{figure}
\begin{lstlisting}
from math import sqrt, floor, ceil
import itertools
import copy
import numpy as np
#########
#This code generates roots of given norm as arrays. 
#We can require integer, or half-integer entries.
##########
def tuples_generator(norm, dim, integer=True):
	if dim == 1:
		if norm == 0:
			return [[0]]

		m = floor(sqrt(norm))
		if not integer:
			m = m + 0.5

		if m*m == norm:
			return [[m], [-m]]
		else:
			return []


	res = []
	m = floor(sqrt(norm))
	if integer:
		ran = range(-m, m+1)
	else:
		if (m+0.5)**2 <= norm:
			ran = range(-m, m+2)
			ran = [e - 0.5 for e in ran]
		else:
			ran = range(-m+1, m+1)
			ran = [e - 0.5 for e in ran]
	
	for e in ran:
		sub_res = tuples_generator(norm - e*e, dim-1, integer=integer)

		for t in sub_res:
			t.append(e)

		res.extend(sub_res)

	return res
\end{lstlisting}
\phantomcaption
\end{figure}
\clearpage
\begin{figure}\ContinuedFloat
\begin{lstlisting}
#This code imposes the lattice condition.

def tuples(norm, dim, integer=True):
	tups = tuples_generator(norm, dim, integer=integer)
	tups = [t for t in tups if sum(t)%2 ==0]
	# tups = [t for t in tups if (sumt(start=0,end=7,tu=t)%2 ==0 and sumt(start=9,end=15,tu=t))]
	return tups

#This code generates vectors of a given norm, where the lattice
condition can be generated by 'chunk',
#to allow us to generate the vectors of E8^3, for instance.
##########
def vectorgenerator(norm, chunk_size, dim):
	if chunk_size == dim:
		res = tuples(norm, dim)
		res.extend(tuples(norm, dim, integer=False))
		return res

	res = []
	for norm_first in range(norm+1):
		fresh = tuples(norm_first, chunk_size)
		fresh.extend(tuples(norm_first, chunk_size, integer=False))

		if len(fresh) > 0:
			sub_res = vectorgenerator(norm - norm_first, chunk_size, dim - chunk_size)
			for part1 in sub_res:
				for part2 in fresh:
					p = copy.deepcopy(part1)
					p.extend(part2)
					res.append(p)

	return res
\end{lstlisting}
\phantomcaption
\end{figure}
\clearpage
\begin{figure}\ContinuedFloat
\begin{lstlisting}
#This code finds all the vectors orthogonal to seven others. 
#For example:
e1 =[2,1,1,1,1,0,0,0]
e2=[1/2,1/2,-1/2,-1/2,-1/2,-1/2,3/2,3/2]
e3=[0,0,1,-1,0,0,0,0]
e4=[0,0,0,1,-1,0,0,0]
e5=[1/2,1/2,-1/2,-1/2,-1/2,-1/2,-1/2,-1/2]
e6=[0,0,0,0,0,0,1,-1]
e7=[0,0,0,0,0,0,0,0]

#This function returns a list of vectors orthogonal to the given vectors.
def orth(n,d):
	vorth = []
	for t in tuples(n, d,True)+tuples(n,d,False):
		if np.dot(t,e1)== 0 and np.dot(t,e2)==0 and
		 np.dot(t,e3)==0 and np.dot(t,e4)==0 and np.dot(t,e5)==1 and np.dot(t,e6)==1 
		 and np.dot(t,e7)==0:
			vorth.append(t)
	return(vorth)
\end{lstlisting}
\caption{Python code to find a lattice vector of given norm orthogonal to a number of given vectors. This code is used to help find a $T_0$ lattice by hand with a convenient form.}
\label{orthcode}
\end{figure}

\begin{figure}
\begin{lstlisting}
#Preliminary definitions
# this evaluates a mod b where a may be a fraction

def modulo(a,b):
    d = denominator(a)
    mod_times_d = (a*d)%(b*d)
    mod = mod_times_d/d
    return mod
   ############# 
  def delta(i,j):
    if i==j:
        return 1
    else:
        return 0
#####################

# input the inner form,  output lattice data about L*/L


def discriminant_form(I):
    r = I.rank()
    lattice = ZZ**r
    dual_lattice_generators = list(I.inverse())
    dual_lattice = span(dual_lattice_generators,ZZ)

    Q = dual_lattice.quotient(lattice)

    Qgens = [Q.gens()[i].lift() for i in range(len(Q.gens()))]

    discriminant_b =
matrix([[modulo(vector(list(Qgens[i]))*(matrix(I)*vector(list(Qgens[j]))),1+delta(i,j))
for i in range(len(Q.gens()))]for j in range(len(Q.gens()))])

    #discriminant_q =
[modulo(vector(list(Qgens[i]))*(matrix(I)*vector(list(Qgens[i]))),2) for
i in range(len(Q.gens()))]

    #print [[k for k in range(Q.invariants()[i])] for i in
range(len(Q.gens()))]

    #print
[[modulo((vector(list(Qgens[i]))*(matrix(I)*vector(list(Qgens[i])))),2)
for k in range(Q.invariants()[i])]  for i in range(len(Q.gens()))]

    discriminant_on_G =
[[modulo(k**2*(vector(list(Qgens[i]))*(matrix(I)*vector(list(Qgens[i])))),2)
for k in range(Q.invariants()[i])]  for i in range(len(Q.gens()))]
    #print discriminant_on_G; we scan over all possible generators: they
are k*generators (should make sure k does not divide group order)

    return (I.determinant(),discriminant_b, discriminant_on_G , Qgens,
Q.invariants())
\end{lstlisting}
\phantomcaption
\end{figure}
\clearpage
\begin{figure}\ContinuedFloat
\begin{lstlisting}
#Notes:
# [0] det of the inner form I of the lattice =  rank of G_L.
# [1] the inner form, evaluated for the generators of G_L, diag elements.
are mod 2 and offdiags are mod 1 (see Nikulin).

# [2] squares of the generators of factors of G_L if we choose different
generators; this is due to the fact that 'different' q_L
# can define the same discriminant form.

# [3] the generators of G, written as rational vectors in Z^m.

# [4] G_L = \prod Z_{ki} ; this is a list of the ki.
\end{lstlisting}
\caption{Code written by A. P. Braun to calculate the discriminant group for a lattice given it's generating matrix.}
\end{figure}

\clearpage
\nocite{*}
\printbibliography[type=article,title={References - Articles}]
\printbibliography[type=book,title={References - Books}]
 
\printbibliography[type=online,title={References - Online resources}]

\end{document}